\documentclass[12pt]{article}

\usepackage{verbatim,color,amsmath,amsthm,amssymb,latexsym}
\usepackage[usenames,dvipsnames]{xcolor}
\usepackage{fancyhdr}
\usepackage[authoryear, sort]{natbib}
\usepackage{graphicx, epstopdf, subcaption}
\usepackage{booktabs}

\def\Section{\hbox{Section}}
\newcommand{\mbf}{\mathbf}

\newcommand{\bbR}{\mathbb{R}}

\newtheorem{thm}{\underline{\bf Theorem}}

\newtheorem{remark}{\underline{\bf Remark}}

\newtheorem{condition}{Condition}
\newenvironment{conbis}[1]
{\renewcommand{\thecondition}{\ref{#1}$'$}%
 \addtocounter{condition}{-1}%
 \begin{condition}}
 {\end{condition}}
\newtheorem{lemma}{\underline{\bf Lemma}}

\def\var{\hbox{var}}
\def\Normal{\hbox{Normal}}
\def\Supp{{\bf Supplementary Material}}
\def\bse{\begin{eqnarray*}}
\def\ese{\end{eqnarray*}}
\def\be{\begin{eqnarray}}
\def\ee{\end{eqnarray}}
\def\bq{\begin{equation}}
\def\eq{\end{equation}}
\def\bse{\begin{eqnarray*}}
\def\ese{\end{eqnarray*}}

\def\th{^{th}}
\def\wh{\widehat}
\def\wt{\widetilde}

\def\boxit#1{\vbox{\hrule\hbox{\vrule\kern6pt  \vbox{\kern6pt#1\kern6pt}\kern6pt\vrule}\hrule}}

\begin{document}

\thispagestyle{empty}
\baselineskip=28pt
{\LARGE{\bf  Nonparametric Bayesian Deconvolution of a Symmetric Unimodal Density}}

\baselineskip=12pt

\vskip 2mm
\begin{center}
Ya Su \\
Department of Statistics, University of Kentucky, Lexington, KY
40536-0082, U.S.A., ya.su@uky.edu\\
\hskip 5mm\\
Anirban Bhattacharya\\
Department of Statistics, Texas A\&M University, College Station, TX
77843-3143, U.S.A., anirbanb@stat.tamu.edu\\
\hskip 5mm \\
Yan Zhang and Nilanjan Chatterjee\\
Departments of Biostatistics and Oncology, Johns Hopkins University, Baltimore, Maryland 21205, U.S.A., yzhan284@jhu.edu and nchatte2@jhu.edu\\
\hskip 5mm\\
Raymond J. Carroll\\
Department of Statistics, Texas A\&M University, College Station, TX 77843-3143, U.S.A. and School of Mathematical and Physical Sciences, University of
Technology Sydney, Broadway NSW 2007, Australia, carroll@stat.tamu.edu\\
\end{center}

\begin{center}
{\Large{\bf Abstract}}
\end{center}
\baselineskip=12pt

We consider nonparametric measurement error density deconvolution
subject to heteroscedastic measurement errors as well as symmetry
about zero and shape constraints, in particular unimodality. The
problem is motivated by applications where the observed data are
estimated effect sizes from  regressions on multiple  factors, where
the target is the distribution of the true effect sizes. We exploit
the fact that any symmetric and unimodal density can be expressed as a
mixture of symmetric uniform densities, and model the mixing density
in a new way using a Dirichlet process location-mixture of Gamma
distributions. We do the computations within a Bayesian context,
describe a simple scalable implementation that is linear in the sample
size, and show that the estimate of the unknown target density is
consistent. Within our application context of regression effect sizes,
the target density is likely to have a large probability near zero
(the near null effects) coupled with a heavy-tailed distribution (the
actual effects). Simulations show that unlike standard deconvolution
methods, our Constrained Bayesian Deconvolution method does a much better job of
reconstruction of the target density. Applications to a genome-wise association study (GWAS) and microarray data
reveal similar results.

\baselineskip=12pt
\par\vfill\noindent
\underline{\bf Some Key Words}: Bayesian methods; Deconvolution; Effect sizes; Shape constraints

\par\medskip\noindent
\underline{\bf Short title}: Deconvolution

\clearpage\pagebreak\newpage
\pagenumbering{arabic}
\newlength{\gnat}
\setlength{\gnat}{22pt}
\baselineskip=\gnat

\section{Introduction}\label{sec1}

In important applied problems, one of which we discuss in Section \ref{sec:data} and the other in the \Supp, data come from a one-dimensional classical measurement
error model $W = X + U$, where the true density of $X$, $f_0(\cdot)$, is assumed to be unimodal and symmetric.
We assume the error has density
$\psi_\sigma(\cdot)$ with mean zero and a scale parameter $\sigma$, details
can be found in the main text, whence the density of $W$,
denoted $p_0$, is the convolution of $f_0(\cdot)$ and $\psi_{\sigma}(\cdot)$.
Given observations $W_1, \ldots, W_n$, our interest lies in estimating
the distribution of $X$ under the given constraints. As part of our
applications, we additionally consider the case where the scales of $U_1, \ldots, U_n$ are heteroscedastic, and denoted as $\sigma_1, \ldots, \sigma_n$.

One of our motivations arises from genome-wide association studies
(GWAS) containing a vast number of single nucleotide polymorphisms
(SNPs) along with a response for a relatively small number of
individuals, where the marginal effect sizes for the SNPs association
with the response are of interest. Let $W_i$ denote the estimated
marginal effect size of the $i$th SNP obtained from a regression of
the response on the $i$th centered and standardized SNP. It can be
shown (see Section \ref{sec:GIANT} for more details) that the true
effect size for the $i$th SNP, $X_i$, can be related to $W_i$ through
$W_i = X_i + U_i$, with the $U_i$ being approximately normally
distributed, but heteroscedastic.

If we treat the true effect sizes $X_i$ as random effects, the sampling distribution of $X_i$ has two key features. First, it makes sense that the effect sizes will be
symmetric about zero and unimodal, and not biased towards being marginally skew. This is the case in our two data applications, where the observed data have almost
zero skewness and are unimodal. Second, in practice, we expect that most of the predictors have very small association with the response, with a handful possibly being
practically significant. This suggests the density should have a sharp peak near zero while possibly being heavy-tailed; for an example of a density satisfying the two
features above, see the blue solid curve in Figure \ref{sim_fig6}. The primary challenge then lies in characterizing the density of $X$ while properly capturing its
expected shape.

There is a rich literature on density estimation in the measurement error context when the measurement error is homoscedastic
\citep{StefanskiDecon1990,CarrollHall1988,Fan1991}, among many others. \cite{delaigle2008density} introduced a deconvoluting kernel technique for the
heteroscedastic measurement error case; see also \cite{sarkar2014density} for a Bayesian approach. However, none of the existing approaches are designed to fulfill
the specific constraints in our case. As a result, we are only able to compare our proposed approach with the general nonparametric kernel deconvolution estimator
\citep{delaigle2008density} in our simulations and real data examples.

In situations without any measurement error, there is some literature on modeling symmetric and unimodal densities. \cite{west1987scale} studied scale-mixtures of
Normals which notably includes the student-$t$ and Laplace families. However, this approach is not fully flexible as there exist symmetric and unimodal densities for
which the underlying mixing functions are not distributions \citep{chu1973estimation}. There are also methods based on Bernstein polynomial basis function where
the shape constraints are preserved under constraints on the coefficients of the basis functions, e.g. \cite{turnbull2014unimodal}. The disadvantages of using Bernstein
polynomial bases are two fold. First, the distribution functions it can characterize exclude those whose support is $(-\infty, \infty)$. Second, the asymptotics of such
shape constrained estimators are not well-studied in the literature even without the measurement error.

In this article, we propose a Bayesian approach for unimodal and symmetric density estimation in the measurement error context. The proposed method is easily
adapted to a heteroscedastic error model, as we will exhibit. A key ingredient of the methodology is a representation theorem for symmetric and unimodal densities
dating back at least to \cite{feller1971introduction}, where it was proved that any unimodal and symmetric density function can be represented by a mixture of uniform
distributions. \cite{brunner1989bayes} adopted this approach and modeled the mixing distribution via a Dirichlet process, which does not yield smooth densities
owing to the almost sure  discreteness of the Dirichlet process. To yield a smooth density, we model the mixing distribution using a Dirichlet mixture of Gamma
distributions, which has large support on the space of smooth densities, and is amenable to scalable posterior computation via an efficient Gibbs sampler we develop
here.

We provide large-sample theoretical support to the proposed
methodology by showing posterior consistency for the observed density
and the latent density. For the observed density of $W$, we borrow
results from recent work \citep{bochkina2017adaptive} where posterior
convergence rates for estimating a density on the positive half-line
were established using Dirichlet location-mixtures of Gamma
distributions. Their setup nicely
serves as a component in our hierarchical model for the density of $W$. While appreciating the value of their theory, the difficulty due to the hierarchical model we
develop and the intrinsic deconvolution problem has not been discussed before and is highlighted in our current work.

We derive a posterior consistency result for the unobserved density of $X$ under a Wasserstein metric. The Wasserstein metric has its origins in the theory of optimal
transportation \citep{villani2008optimal} and has recently been found suitable for studying convergence of mixing measures in deconvolution problems
\citep{nguyen2013convergence, gao2016posterior,scricciolo2018bayes} . These papers consider a Dirichlet process mixture type of model where the mixing
distribution is discrete and needs to satisfy some conditions, see Section \ref{sec4} for a discussion on their conditions . A key ingredient of our theory is the
development of a new inversion inequality which relates the convergence of the observed/mixture density to that of the unobserved/mixing density. The idea of using
inversion inequalities in the Bayes literature is fairly new, with only a few instances of such results, e.g.,  \cite{nguyen2013convergence}, \cite{scricciolo2018bayes}.
However, existing inequalities can not be applied directly to our case, necessitating a new inversion inequality to fit our needs.

Section \ref{sec:model} gives the Bayesian model leading to our methodology, while Section \ref{sec:theory} states asymptotic results. Section \ref{sec:algorithm}
describes our algorithm and Section \ref{sec:simulation} presents some of the many simulations we have conducted. Section \ref{sec:data} presents an analysis of a
genome-wide association study, and shows that our methodology is able to capture the mixture distribution we expect to see in the data as described above. Section
\ref{discuss} gives concluding remarks. \Supp\ includes additional data analysis of a microarray experiment.

\section{Model Specification}\label{sec:model}

Throughout our paper, $\psi(\cdot)$ denotes a symmetric
unimodal density on the real line which specifies our family of error
distributions. We further denote by $\psi_\sigma(\cdot)$ the corresponding
scale family: $\psi_\sigma(t) = (1/\sigma) \, \psi(t/\sigma)$
for $\sigma > 0$. Finally, $\Psi_{\mu, \sigma}(\cdot)$ denotes the distribution function with density (in $t$) given by
$(1/\sigma) \psi \{(t-\mu)/\sigma$\}.

Since $W = X + U$, the true density $p_0(\cdot)$ of $W$ has the form
\begin{equation} \label{eq:2}
p_0(w) = \int \psi_\sigma(w - x) f_0(x) dx,
\end{equation}
where the true density of $X$, $f_0(\cdot)$, has a unimodal and
symmetric shape.  If $f_0$ is continuous with finite derivative $f_0'(x)$ for all $x$, then it is well-known \citep{feller1971introduction}  that there exists a density
$g_0(\cdot): \mathbb{R}^+ \rightarrow \mathbb{R}^+$, where $\mathbb{R}^+ = [0, \infty)$, such that
\begin{eqnarray} \label{eq:11}
f_0(x) = \int (2 \theta)^{-1} I_{(-\theta \leq x \leq \theta)} g_0(\theta) d\theta.
\end{eqnarray}
In other words, any symmetric and unimodal density is a mixture of
symmetric uniforms. Given our motivating application, it is natural to
assume in addition that $f_0(\cdot)$ is finite at zero, which ensures the
finiteness of $p_0(\cdot)$. The finiteness of $f_0(0)$ can in turn be
ensured by assuming that $g_0(0) = 0$. Our parameter space for $g_0(\cdot)$ thus consists of all densities on the positive half-line $\mathbb{R}^+$ satisfying $g_0(0)
= 0$.

In the deconvolution literature, two types of error distributions,
ordinary-smooth and super-smooth, are commonly studied. By definition,
a density is ordinary-smooth or super-smooth if the tail of its
Fourier transform decays to zero at polynomial rate or exponential
rate, respectively.
For our theoretical analysis and simulation studies, we pick one
distribution from each class, namely the Normal and Laplace
distributions. When presenting the theory we illustrate the Normal
error case first, while the results for the Laplace error distribution
are studied in a separate section. A similar strategy has been taken
with the proofs. Furthermore in a more complicated situation when only
the type (ordinary-smooth or super-smooth) is known, we point out the
possibility of modeling the error distribution using mixtures of
Normal/Laplace distributions prior; see \cite{SarkarMultivariate} for
an instance of the former.

We build our Bayesian model in a hierarchical structure as the true densities, that is, the candidate densities $p(\cdot)$, $f(\cdot)$ and $g(\cdot)$ are defined in a similar way as in (\ref{eq:2}) and (\ref{eq:11}). In particular, given the representation \eqref{eq:11}, the problem of modeling $f(\cdot)$ equivalently reduces to creating a flexible model for $g(\cdot)$. Recall that $g(\cdot)$ is supported on $\mathbb{R}^+$. We model $g(\cdot)$ using a Dirichlet process location-mixture of Gamma distributions, which has large support \citep{bochkina2017adaptive} on densities supported on $\mathbb{R}^+$, and is easy to implement in a Bayesian framework. Specifically, we reparameterize a Gamma density by its shape $z$ and mean $\mu$ as parameter pairs. Denote $g_{z, z/\mu}$ to be a Gamma density with shape $z$ and rate $z/\mu$; we use $\mbox{Ga}(z, z/\mu)$ to denote the corresponding probability distribution. We assume a Dirichlet process prior \citep{ferguson1973bayesian} on the distribution of $\mu$ and another prior $\Pi_z$ on $z$. With these ingredients, our hierarchical Bayesian model is
\begin{eqnarray*}
  &W_i|X_i \sim \Psi(X_i, \sigma);  \quad X_i|\theta_i \sim \text{Unif}(-\theta_i, \theta_i);  \quad \theta_i|z, \mu \sim \text{Ga}(z, z/\mu); \\
  &\mu|P_\mu \sim P_\mu; \quad P_\mu|m, D \sim \text{DP}(m,D); \quad z \sim \Pi_z,
\end{eqnarray*}
where $\text{Unif}(\theta_1, \theta_2)$ is a Uniform distribution on the interval $[\theta_1,\theta_2]$ and $\text{DP}(m, D)$ denotes a Dirichlet process with
concentration parameter $m$ and base probability measure $D$. The hyperparameters are $m$ and other possible parameters for specification of $D$ and $\Pi_z$.

Using the stick-breaking representation \citep{sethuraman1994constructive} for the Dirichlet process, the model-prior for $g(\cdot)$ can be represented as
\begin{align*}
& g(x) = \int \{ \hbox{$\sum_{h=1}^{\infty}$} \nu_h\,\mbox{Ga}(x \mid z, z/\mu_h) \} \, \Pi_z(dz), \\
& \nu_h = \nu_h^\ast\, \hbox{$\prod_{\ell < h}$} (1 - \nu_\ell^\ast), \quad \nu_\ell^\ast \sim \mbox{Beta}(1, m), \quad \mu_h \sim D,
\end{align*}
where $\mbox{Ga}(x \mid z, z/\mu_h)$ denotes the $\text{Ga}(z, z/\mu)$ density evaluated at $x$. For numerical computation, we use a finite Dirichlet approximation
\citep{ishwaran2002exact} to the Dirichlet process in our simulations and data examples.

\section{Theoretical Analysis}\label{sec:theory}

\subsection{Goal and Background}\label{sec2}

In this section, we provide theoretical support to our method in terms of posterior consistency for the observed and latent densities. Specifically, we show that the posterior distribution for $p(\cdot)$ and $f(\cdot)$ increasingly concentrates on arbitrarily small neighborhoods of the true densities $p_0(\cdot)$ and $f_0(\cdot)$, respectively, as the sample size increases.

We follow the general procedures in
  \cite{ghosal2000convergence} of establishing posterior contraction
  theory and make substantial modifications to adapt to the
  hierarchical model considered in this paper. We begin with a basic
  model with no measurement error and then build the theory towards
  its measurement error counterpart, allowing multiple layers of
  mixture in the latter case. Another novelty of the current approach
  is its ability to work with $X$ having a continuous density with
  infinite support, as opposed to a discrete density with finite
  support considered in \cite{nguyen2013convergence}. 
  This is achieved by a mixture model with a mixing
      distribution modelled by a Dirichlet Process mixture of Gamma
      distributions. We obtain some preliminary
results on this layer from \cite{bochkina2017adaptive}. An inversion
inequality is derived that bridges our theory from $p(\cdot)$ to $f(\cdot)$.

We list some key definitions and notation in this section. Let $\mu$
and $\nu$ be two probability measures defined on a metric space with
metric $d$. If $\mu$ and $\nu$ both have finite $p$th moments, the
$p$th Wasserstein distance \citep{villani2008optimal}, denoted
$W_p(\mu, \nu)$, is defined as $W_p^p(\mu, \nu) = \inf_{\phi \in
  \Gamma(\mu, \nu)} \int d^p(x, y) d\phi(x, y)$, where $\Gamma(\mu,
\nu)$ represents the collection of all joint measures with marginal
measures $\mu$ and $\nu$. We consider the metric space $\mathbb{R}$
with the Euclidean distance $d(x, y) = |x - y|$. For any two densities
$p_1(\cdot)$ and $p_2(\cdot)$ on $\mathbb{R}$, $W_p(p_1, p_2)$ is the
same as $W_p(P_1, P_2)$ where $P_1$ and $P_2$ are the cumulative
distribution functions corresponding to $p_1(\cdot)$ and $p_2(\cdot)$,
respectively. Another distance metric between two probability
densities $p_1(\cdot)$ and $p_2(\cdot)$ is the Hellinger distance,
$h(p_1, p_2) = (1/2) \int \{p_1(x)^{1/2} - p_2(x)^{1/2}\}^2 dx$. The
Hellinger distance is widely used in the Bayesian asymptotics
literature for quantifying posterior consistency or convergence of
densities. The notation $\Pi_n (A_n |W_1, \ldots, W_n)$ stands for a posterior probability of an event $A_n$ given the observations $W_1, \ldots, W_n$.

To make notation simpler, from now on, we assign an overall symbol $P_0$ for probability or expectation under the true distribution of the corresponding variable, e.g., $P_0(W > s)$ or $P_0(X > s)$ mean the probability that $W > s$ or $X > s$ under the true $p_0$ or $f_0$ respectively. Also, $a_n \lesssim b_n$ $(a_n \gtrsim b_n)$ means that there exists a positive constant $C$ such that $a_n/b_n \leq C$ $(a_n/b_n \geq C)$ for all $n$. In addition, $a_n \asymp b_n$ if and only if $a_n \lesssim b_n$ and $a_n \gtrsim b_n$, $a \vee b=\max(a, b)$, $a \wedge b=\min(a, b)$. Finally, $\lceil a \rceil$ denotes the smallest integer that is greater than or equal to $a$.

\subsection{Posterior Consistency for the Observed Density}\label{sec3}

This section gives a theorem on the posterior convergence rate for $p(\cdot)$. Our conditions are mainly at the layer of $g(\cdot)$, which is modelled as a Dirichlet
location-mixture of Gamma distributions. We will give the conditions followed by some interpretations on these conditions and then state the theorem.

\begin{condition}\label{con1}{\rm
We adopt a function space for $g_0(\cdot)$, $\mathcal{M}\{L(\cdot), \varpi, C_0, C_1, e, \Delta\}$, which contains a set of  density functions $q: \bbR^+ \rightarrow
[0,\infty)$ which satisfy that there exists $L(\cdot) > 0, \varpi \geq
0, C_0 > 0 , C_1 > 0, e > 0$ and $\Delta$ that for all $\theta \in
\bbR^+$, $\phi > - \theta$ and $|\phi| \leq \Delta$,
 \begin{eqnarray*}
&&   |q(\theta + \phi) - q(\theta)| \leq L(\theta) |\phi| (1 + |\phi|^\varpi); \enspace q(\theta) \leq C_0; \\
&&    \int_0^\infty \{(1 + \theta^{\varpi})\theta L(\theta)/q(\theta)\}^2 q(\theta) d\theta \leq C_1.
  \end{eqnarray*}
  }
\end{condition}

 \begin{condition}\label{con2} {\rm For some $\rho_1 > 2$, $\int_x^\infty \theta^4 g_0(\theta) d\theta \leq C(1 + x)^{-\rho_1+2}$.}
  \end{condition}

\begin{condition}\label{con4}{\rm
   (i) The prior on $P_\mu$ is $\text{DP}(m, D)$, where $D$ has a positive and continuous density $d(\cdot)$ on $\mathbb{R}^+$ satisfying that for some $0 <
   a_0^\prime \leq a_0$ and $0 < a_1^\prime \leq a_1$,
  \begin{align*}
    \exp(- x^{- a_0}) \lesssim d(x) \lesssim \exp(- x^{-a_0^\prime}) \text{ as } x \rightarrow 0; \\
    \exp(- x^{a_1}) \lesssim d(x) \lesssim \exp(- x^{a_1^\prime}), \text{ as } x \rightarrow \infty.
  \end{align*}
  (ii) The prior on $z$, $\Pi_z$, has support (1,
  $\infty$). For constants $c \geq c^\prime > 0$, $c_0 > 0$ and
  $\rho_z \geq 0$,
  \begin{align*}
   &\Pi_z([x, 2x]) \gtrsim \exp\{-c\sqrt{x}(\log x)^{\rho_z}\}, \enspace \Pi_z([x, \infty)) \lesssim \exp\{-c^\prime \sqrt{x}(\log x)^{\rho_z}\} \text{ as } x \rightarrow
   \infty, \\
   &\Pi_z((1, x]) \lesssim (x - 1)^{c_0} \text{ as } x \rightarrow 1.
  \end{align*}
  }
\end{condition}

For notational simplicity, we drop the arguments and only use
$\mathcal{M}$ to denote the space of densities in Condition
\ref{con1}. Similar function spaces with additional smoothness
assumptions have been used by \cite{bochkina2017adaptive}; we do not
make such smoothness assumptions here. The conditions are typical in
the literature on Bayesian density estimation. A
density satisfying Condition \ref{con1} and Condition \ref{con2} can
be well approximated by a mixture of Gamma distributions which
facilitates finding a KL divergence neighbourhood around the true
observed density $p_0(\cdot)$. When the error distribution is Laplace,
Condition \ref{con2} is slightly relaxed, see Condition \ref{con3}
below. Condition \ref{con4} (i) is on the base measure of Dirichlet
process and agrees with that in \cite{shen2013adaptive} except that
the support is on $(0, \infty)$ instead of $(-\infty,
\infty)$. Condition \ref{con4} mainly controls the prior thickness of
the sieve space upon which the inversion inequality in Section \ref{sec4}
can be derived. \cite{bochkina2017adaptive} showed Condition
\ref{con1} is satisfied by Weibull, folded Student-t and Frechet-type
densities.  Condition \ref{con4} (ii) holds, for example, if $\sqrt{z}$ has a Gamma
prior.

Clearly, the prior is hierarchical, Condition \ref{con1} and Condition
\ref{con2} are imposed on $g_0(\cdot)$ which is free of shape constraints
except that it is a density on the positive half line. It is generally
difficult to do the other way around, that is, impose conditions on
$f_0(\cdot)$ and identify its corresponding properties on
$g_0(\cdot)$. However, we can verify these conditions under some special
cases. When $f_0(\cdot)$ is a Normal density with mean zero and
standard deviation $\sigma$, $g_0(\theta) = C (\theta/\sigma)^2
\exp\{-(\theta/\sigma)^2\}$ which belongs to a Weibull family of
distributions. Therefore Condition \ref{con1} is met. Condition
\ref{con2} holds for arbitrarily large $\rho_1$. When $f_0(\cdot)$ is a
t-distribution with degrees of freedom $\nu$, $g_0(\theta) = C
\theta^2 (1 + \theta^2)^{-(\nu + 3)/2}$ which is an Inverse Beta
distribution. Condition \ref{con1} can be verified by similar
arguments in \cite{bochkina2017adaptive} for a folded Student-t
density since only the tail
behavior of its derivatives matters. Condition \ref{con2} holds when
$\nu > 4$ with $\rho_1 = \nu - 2$.

\begin{thm}\label{thm_p}
  {\rm Fix $\epsilon > 0$. Under Conditions \ref{con1}--\ref{con4}, for any $M > 0$ large enough,
 \begin{equation*}
 \lim_{n \rightarrow \infty} \Pi_n (\{p: h(p, p_0) > M \epsilon\}|W_1, \ldots, W_n) = 0 \text{ almost surely.}
\end{equation*}
}
\end{thm}

\begin{proof}{\rm
To prove Theorem \ref{thm_p}, we shall exhibit a sequence $\epsilon_n \to 0$ such that
 \begin{equation*}
\lim_{n \rightarrow \infty} \Pi_n (\{p: h(p, p_0) > M \epsilon_n\}|W_1, \ldots, W_n) = 0 \text{ almost surely.}
\end{equation*}
To prove the assertion in the above display, it follows from \cite{ghosal2000convergence} that
the desired result holds as long as there exists a sequence of compact subsets $\{\mathcal{F}_n\}$ in the space where $p(\cdot)$ resides and a sequence $\{
\wt{\epsilon}_n \}$ with $\wt{\epsilon}_n \leq \epsilon_n$ and $\lim_{n \rightarrow \infty} n \wt{\epsilon}_n^2 = \infty$ such that
  \begin{eqnarray}
    \label{eq:6}  \log N(\epsilon_n, \mathcal{F}_n, h) &\leq& c_1 n \epsilon_n^2 ; \\
    \label{eq:7}  \Pi(\mathcal{F}_n^c) &\leq& c_3 \exp\{-(c_2 + 4)n \wt{\epsilon}_n^2\} ; \\
    \label{eq:8}  \Pi\{p: P_0 \log (p_0/p) &\leq& A \wt{\epsilon}_n^2, P_0\{\log (p_0/p)\}^2 \leq A \wt{\epsilon}_n^2 \log n \} \geq \exp(-c_2 n \wt{\epsilon}_n^2),
  \end{eqnarray} for some positive constants $c_1, c_2, c_3, A$, and $N(\epsilon_n, \mathcal{F}_n, h)$ is the $\epsilon_n$-covering number of $\mathcal{F}_n$ relative
  to the Hellinger distance. Equations (\ref{eq:6}) and (\ref{eq:7}) are entropy and prior mass conditions on the sieve space and (\ref{eq:8}) is referred to as the prior
  concentration condition. Equation (\ref{eq:8}) is a slight variation compared to the original prior concentration condition in \cite{ghosal2000convergence}; see
  \cite{bochkina2017adaptive}.

  In Appendix \ref{sec:app:6}, the details for deriving equations (\ref{eq:6}), (\ref{eq:7}) and (\ref{eq:8}) are provided for $\wt{\epsilon}_n^2 = \epsilon_n^2 \asymp n^{-2/(2B + 3)}(\log n)^{(2B + 2)/(2B + 3)}$ and an appropriate sieve space $\mathcal{F}_n$. The constant $B$ in $\epsilon_n$ is determined by the constants $\rho_1$, $a_0$ and $a_1$ in Condition \ref{con2} and \ref{con4} (i).  }
\end{proof}

\subsection{Posterior Consistency for the Latent Density}\label{sec4}

We now establish that the posterior distribution for the latent density $f(\cdot)$ increasingly concentrates around the true density $f_0(\cdot)$. To show such a result, we build an inversion inequality which harnesses the consistency of the observed density $p(\cdot)$ derived above to prove consistency for the latent density $f(\cdot)$.  A few previous instances of inversion inequalities can be found in the recent literature.  Theorem 2 of \cite{nguyen2013convergence} relates the Wasserstein distance between the mixing distributions with the total variation of the mixture density, but it requires the mixing distribution to reside on a finite support or have bounded $s > 2$ moment. \cite{scricciolo2018bayes} makes use of an inversion inequality to establish the convergence rate of the Bayes estimator for the mixing density; one of the key requirements on the mixing distribution is that it has a bounded moment generating function on some interval containing $[-1,1]$. However, there does not exist an inversion inequality that can be directly applied to our problem, where the mixing density $f(\cdot)$ has unbounded support and there is no way to bound the moment generating function on any interval containing $[-1,1]$ for all $f(\cdot)$ in a sieve space. In Appendix \ref{sec:app:2}, we prove the next Lemma that relates the convergence of $f(\cdot)$ to $f_0(\cdot)$ under the Wasserstein metric, $W_2(f, f_0)$, and the $L_1$ distance between $p(\cdot)$ and $p_0(\cdot)$.

\begin{lemma}\label{inv_ine}
  {\rm
  On the sieve $\mathcal{F}_n$ in Theorem \ref{thm_p}, when $\rho_1$
  and $a_1^\prime$ (see Condition \ref{con2} and Condition \ref{con4}) are large enough,
  \begin{equation*}
    W_2^2(f, f_0) \lesssim \{-\log(\|p - p_0\|_1)\}^{-1}.
  \end{equation*}
  }
\end{lemma}

\begin{remark}\label{rem4}
{\rm For any two densities $p_1$, $p_2$, $\|p_1 - p_2\|_1/2 \leq h(p_1, p_2) \leq \|p_1 - p_2\|_1^{1/2}$. The conclusion of Lemma \ref{inv_ine} can be equivalently
stated as $W_2^2(f, f_0) \lesssim [-\log\{h(p, p_0)\}]^{-1}$.}
\end{remark}

\begin{thm}\label{thm_f}
  {\rm
  Fix $\varepsilon > 0$. Under the Conditions in Theorem \ref{thm_p} and Lemma \ref{inv_ine}, for any $M > 0$ large enough,
 $\lim_{n \rightarrow \infty} \Pi_n [{f: W_2(f, f_0) > M \varepsilon}|W_1, \ldots, W_n] = 0 \text{ almost surely}.$
}
\end{thm}

\begin{proof}
Theorem \ref{thm_f} follows from Theorem \ref{thm_p} and Lemma \ref{inv_ine}.
\end{proof}

\begin{remark}
{\rm Theorem \ref{thm_f} states that the posterior consistency of
  $f(\cdot)$ in the $W_2$ metric as a result of the presence of the
  $W_2$ metric in the inversion inequality in Lemma \ref{inv_ine}. In
  fact, the proof of Lemma \ref{inv_ine} can be extended to $W_k$ for
  any $k \geq 1$, which in turn would imply posterior consistency in
  any $W_k$ metric. To the best of our knowledge, technical difficulties
  exist in order to derive Lemma \ref{inv_ine} for the $L_1$ metric between
  $f(\cdot)$ and $f_0(\cdot)$. The difficulties lie in finding a {\it uniform} upper bound for the $L_1$ distance between functions in the sieve space and its convolution with the molifier. Whereas if Wasserstein distance (of order $2$) is in use, such an upper bound is simply the second moment of the molifier. This is probably the
  hurdle if one wants to
  establish posterior contraction theory in $L_1$ distance for the mixing density
  without restricting oneself on special cases of the mixing
  density.
}
\end{remark}

\subsection{Theory when the error has a Laplace distribution}\label{sec5}

All theorems and Lemmas in Section~\ref{sec3} and Section~\ref{sec4}
can be derived when the measurement error has a Laplace distribution
under a relaxation of Condition \ref{con2}. We state the condition and
theorems whenever changes are met.

\begin{conbis}{con2}\label{con3}{\rm
 For some $\rho_1 > 0$, $\int_x^\infty \theta^2 g_0(\theta) d\theta \leq C(1 + x)^{-\rho_1}$.
}
\end{conbis}

It can be inferred that Condition \ref{con3} holds for $\rho_1 > 2$
assuming Condition \ref{con2}. The
statement in Theorem \ref{thm_p} holds under Condition \ref{con1}, \ref{con3} and \ref{con4}.

\begin{lemma}\label{inv_ine_lap}
  {\rm
  On the sieve $\mathcal{F}_n$ in Theorem \ref{thm_p}, when
  $\rho_1$ and $a_1^{\prime}$, see Condition \ref{con3}, and \ref{con4}
  (i) are large enough, there exists a $\nu > 0$ depending on
  $\rho_1$ and $a_1^\prime$ such that
  \begin{equation*}
    W_2^2(f, f_0) \lesssim \|p - p_0\|_1^\nu.
  \end{equation*}
  }
\end{lemma}

Theorem \ref{thm_p} and Lemma \ref{inv_ine_lap} together imply that
Theorem \ref{thm_f} holds.

 The proofs are along the lines of their correspondence to the Normal error
case. They are in Appendix \ref{sec:app:6} with only the differences presented.

\section{Algorithm}\label{sec:algorithm}

To ease computational complexity, we follow standard practice by
approximating the Dirichlet process mixture prior with a finite
mixture of Gamma distributions with $K$ components where K is large,
with a specific Dirichlet prior on the mixture probabilities
\citep{ishwaran2002exact}. It is trivial to implement our procedure
for the infinite mixture using the slice sampler of
\cite{kalli2011slice}; however we prefer the finite Dirichlet due to
its substantially better mixing behavior for our multi-layered
hierarchical model. Our theoretical results in Section
\ref{sec:theory} were developed for the Dirichlet location-mixture of
Gamma priors on $g(\cdot)$, where only the mean parameter is mixed
over.  For flexibility, we adopt a mixture on both the shape and rate
parameters for our numerical implementation.  The conditions on the
priors for these parameters become less stringent because the number
of such parameters is finite. We select these priors among some
popular choices. Specifically, our hierarchical Bayes model for
subsequent implementations is as follows. Let $i$ denote the index for
subject, and $k$ be the index for the $k$th component, for all
$i = 1, \ldots, n$, $k = 1, \ldots, K$. Let $t > 1$ denote a fixed
constant. Then,
\begin{eqnarray*}
  &&(W_i|X_i) \sim \Psi(X_i, \sigma_i); \enspace
  (X_i|\theta_i) \sim \text{Unif}(-\theta_i, \theta_i); \enspace
  (\theta_i|Z_i = k, \alpha_k, \beta_k) \sim \text{Ga}(\alpha_k,
  \beta_k); \\
  &&P(Z_i = k|p_1, \ldots, p_K) = p_k; \enspace
  (\alpha_k|\lambda, t) \sim \text{Expon}(\lambda; t, \infty); \enspace
  (\beta_k|\Xi_1, \Xi_2) \sim \text{Ga}(\Xi_1, \Xi_2); \\
  &&(p_1 \ldots, p_K) \sim \text{Dirichlet}(m/K, \ldots, m/K),
\end{eqnarray*}
where $\text{Dirichlet}(\gamma_1, \ldots, \gamma_K)$ denotes a Dirichlet distribution with parameters $\gamma_1, \ldots, \gamma_K$, $\text{Expon}(\lambda; \ell,
u)$ denotes an exponential distribution with parameter $\lambda$ truncated at $(\ell, u)$. The paragraph above Theorem \ref{thm_p} points out the reason for
truncating $\alpha_k$. The set of hyperparameters is $(\lambda, t, \Xi_1, \Xi_2, K, m)$.

Denote the set of all variables and hyperparameters given above as
\begin{equation*}
  \mathbf{\Omega} = (\{W_i\}_{i=1}^n; \{X_i\}_{i=1}^n; \{\theta_i\}_{i=1}^n; \{Z_k\}_{k=1}^K; \{\alpha_k\}_{k=1}^K;
\{\beta_k\}_{k=1}^K; \{p_k\}_{k=1}^K; \lambda, t, \Xi_1, \Xi_2, K, m).
\end{equation*}
For ease of notation, let $\mathbf{\Omega}_{-\zeta}$ be all variables in $\mathbf{\Omega}$ but excluding $\zeta$. For $k = 1, \ldots, K$, let $r_k = \sum_i I_{(Z_i = k)}$ be the total number of individuals that fall into group $k$ and $s_k = \sum_i \theta_i I_{(Z_i = k)}$ be the summation of the $\theta_i$ from the $k$th group. To sample from the posterior distribution of $\Omega$, we use a Gibbs sampler for all parameters other than the $\alpha_k$, combined with a Metropolis-Hastings within Gibbs for the $\alpha_k$. The posterior full-conditional distributions are
\begin{eqnarray*}\label{eq:posterior}
  (X_i | \mathbf{\Omega}_{-X_i}) &\sim& \Psi(W_i, \sigma_i; -\theta_i, \theta_i);\\
  (\theta_i | \mathbf{\Omega}_{-\theta_i}) &\sim& \text{Ga}(\alpha_{Z_i} - 1,
                                                  \beta_{Z_i}; |X_i|, \infty);\\
  P(Z_i = k|\mathbf{\Omega}_{-Z_i}) &\propto& \Gamma(\alpha_k)^{-1} p_k (\beta_k \theta_i)^{\alpha_k}
                                              \exp(-\beta_k \theta_i);
  \\
  (p_1, \ldots, p_K | \mathbf{\Omega}_{- \{p_1, \ldots, p_K\}}) &\sim& \text{Dirichlet}(m/K + r_1, \ldots, m/K + r_K);\\
  (\beta_k | \mathbf{\Omega}_{-\beta_k}) &\sim& \text{Ga}(\Xi_1 + \alpha_k r_k, \Xi_2 + s_k);
  \\
  (\alpha_k | \mathbf{\Omega}_{-{\alpha_k}}) &\propto&
                                                       \Gamma(\alpha_k)^{-r_k} \exp\{-
                                                       \alpha_k(\lambda -
                                                       r_k \log \beta_k -
                                                       \scalebox{0.9}{$\sum$}_i
                                                       \log(\theta_i) I_{(Z_i = k)})\}.
\end{eqnarray*}
The symbol $\Psi(\mu, \sigma; \ell, u)$ denotes the distribution $\Psi(\mu, \sigma)$ truncated at $(\ell, u)$. Meanwhile $\text{Ga}(\alpha, \beta; \ell, u)$ corresponds to a Gamma distribution with parameters $(\alpha, \beta)$ truncated at $(\ell, u)$. Since the posterior distribution of $\alpha_k$ does not belong to a standard family, we implement a Metropolis-Hastings algorithm within the Gibbs sampler to update the $\alpha_k$. We use a Gamma proposal distribution; specifically, $\wt{\alpha}_k \sim \text{Ga}(2, 2/\alpha_k; t, \infty)$, and we accept the proposed $\wt{\alpha}_k$ or keep the original $\alpha_k$ according to the general Metropolis-Hastings rule. The proposal distribution is truncated to reflect the prior assumption on $\alpha_k$.

For all of our simulations presented, we treat the error variances $\sigma_i^2$ for all $U_i$ as known: this is reasonable in our examples, and often used in the
standard deconvolution theory. The default selected values for hyperparameters are $\lambda = 2, t = 2.5, \Xi_1 = 1, \Xi_2 = 4, K = 8, m = 20$. Sensitivity analysis
showed little sensitivity to different choices of the hyperparameters. The marginal density for $X$, our estimator, is computed as the average value of the marginal
density at each MCMC iteration. We name the method as Bayes density deconvolution with shape constraint estimator (Constrained Bayes Deconvolution).

Our Constrained Bayes Deconvolution method is easily seen to be scalable in that it is linear in the sample size, and indeed in Section \ref{sec:GIANT} it is show to
be able to handle sample size of nearly $10^6$: it is written in R with use of the package RCPP.

\section{Simulations}\label{sec:simulation}

\subsection{Overview}\label{sec5.1}

We conducted simulations for two distinctly different problems. In the
first, the target density for $X$ has a standard t-distribution with 5
degrees of freedom. In the second, related to our examples, $X$ has a
density that is a mixture of (a) t random variables with 5 degrees of
freedom; and (b) values with mean zero and very small variability. In
addition, for each of (a) and (b), we consider the case of
homoscedastic and heteroscedastic measurement errors generated from either the Normal or the Laplace distributions.

Case (b) is the important one for us given the type of data we
want to analyze, while Case (a) is simply meant to show that we
are competitive with the standard method, namely the kernel density
deconvolution estimator, in standard problems. The kernel estimator
has two versions depending on whether the measurement errors are
homoscedastic or heteroscedastic. The plug-in bandwidth, which
minimizes the asymptotic mean integrated squared error, is chosen for
this estimator in comparison with our method, see
\cite{delaigle2008density}. The R package, deconvolve,
published on Github implements the kernel density deconvolution estimator.

In each design of the simulation we generated data with sample sizes
$n = 1,000$, $5,000$, each repeated with $100$ simulated data
sets.

We compute posterior samples of the density across the MCMC steps and the estimated density is
obtained as the mean of these posterior samples. The estimated
densities and the true density are compared via the  square root of the integrated squared error (ISE),
the integrated absolute error (IAE) and the Wasserstein
distance ($W_2$) for each simulated data set. An overall summary is
given in Section \ref{sec5.4}.

\subsection{When $X$ has a t-distribution With 5 Degrees of Freedom}\label{sec5.2}

We generated observations by $W_i = X_i + U_i$, $X_i$ has a $t$ distribution
with $5$ degrees of freedom. In the case of homoscedastic error,
the variance of $U$ is equal to the variance of $X$, specifically, $\mbox{Var}(U_i) = 1.66$.
In the heteroscedastic case,  $\mbox{Var}(U_i) = (1 + X_i/4)^2$, with the variance of $X$ being $1.5$ times the mean of $\mbox{Var}(U_i)$. In all cases,
the observations are subject to substantial measurement error.
The estimated densities are displayed in Figure
\ref{sim_fig1} -- Figure \ref{sim_fig4}. The numerical comparisons for
our Constrained Bayes Deconvolution method and the Kernel method are given in Table
\ref{tab2} -- Table \ref{tab4}.

\subsection{When $X$ has a Tight Peak Around Zero}\label{sec5.3}

The setting in this section is designed for cases when the distribution of $X$ has a large probability clustered near zero, as we expect in our examples. One way to do this is through a mixture structure, assuming that the density of $X$ has a component that is tightly concentrated at zero and another component from a standard density. We implement a mixing of a $\Normal(0, \sigma_{00}^2)$ for the first component and a $t$-distribution with $5$ degrees of freedom for the second component, with mixing probabilities $0.8$ and $0.2$ respectively. We choose the small value $\sigma_{00} = 0.2$ so that the mixing density has a very sharp peak around zero. For $\sigma_{00} = 0.2$ $\var(X)=0.37$.

In this case, when the true density puts a high concentration around
zero, in addition to the usual global metrics IAE, ISE and $W_2$, it is interesting to study how well an estimated density can capture the probability greater than, in
absolute value, $3$ times the standard deviation of the ``tight peak'' component. With a small abuse of notation, in the following, ``Exceedance'' is defined as the
absolute difference between the exceedance probability under the estimated density and that under the true density.

In the case of homoscedastic error,
$\mbox{Var}(U_i) = 0.36$, such that the variance of $U$ is equal to the variance of $X$.
We implement the heteroscedastic case by adjusting an appropriate form
for $\mbox{Var}(U_i)$ in Section \ref{sec5.2} such that the mean of
$\mbox{Var}(U_i)$ is more than the variance of $X$, specifically,
$\mbox{Var}(U_i) = (0.75 + X_i/4)^2$. Again in all cases,
the observations are subject to substantial measurement error.
The estimated densities are displayed in Figure
\ref{sim_fig6} -- Figure \ref{sim_fig8}. The numerical comparisons for
our Constrained Bayes Deconvolution method and the Kernel method are given in Table
\ref{tab6} -- Table \ref{tab7}.

\subsection{Conclusions from the Simulations}\label{sec5.4}

For both the simulations in Section \ref{sec5.2} and Section \ref{sec5.3}, with either homoscedastic or heteroscedastic error, we observe that under the global metrics ISE and IAE, large gains in efficiency are achieved with our Constrained Bayes Deconvolution estimator over the deconvoluting kernel estimator across all choices of sample size. Also, from the figures and tables of Section \ref{sec5.3}, with either homoscedastic or heteroscedastic error, the Constrained Bayes Deconvolution estimator performs much better in capturing the peak as well as the tail behavior, from both a visual check and the Exceedance metric. Lastly, the kernel deconvolution estimator gives a biased peak for our sample sizes when the errors are heteroscedastic.

\section{Genome Wide Association Applications}\label{sec:data}

\subsection{Background}\label{sec6.1}

In this section, we describe the results of a genome-wide association study (GWAS) that is particularly appropriate. In the \Supp, we also describe results from a
microarray experiment, which reaches similar conclusions.

\subsection{Height data}\label{sec:GIANT}

Our data come from a genome-wide association study for height \citep{allen2010hundreds}. The study data we have involves 133,653 individuals, and each individual in our data set has 941,389 SNPs that were measured. The goal of the study was to understand which SNPs were related to height, either positively or negatively.  Because of the relative rareness of traits that affect height, the simulation of Section \ref{sec5.3} is particularly relevant.

The data we have access to are regression coefficients of standardized heights, $Y_k$ say, on standardized SNPs for SNP $i$, $Z_{ik}$ say, and are thus estimated effect sizes. If we regress the $Y_k$ on the $Z_{ik}$, it is easy to see that if the true effect size is $X_i=\beta_i$, the estimated effect size is $W_i=\wh{\beta}_i$, which, because of the sample size involved, is approximately normally distributed with mean $\beta_i$ and measurement error $U_i = \Normal(0,\sigma_i^2)$, where $\sigma_i^2 = \sigma^2_{i\epsilon}/n$, where $n$ is the sample size and $\sigma^2_{i\epsilon}$ is the regression variance of the $Y_k$ on the $Z_{ik}$. Clearly, because of the sample size and the division by $n$, $\var(U_i) = \sigma_i^2$ is well-estimated and thus essentially known, but heteroscedastic.

For our Constrained Bayes Deconvolution estimator, we run 5000 MCMC iterations using the same hyperparameters used in the simulation section. There was a difficulty with the deconvoluting kernel density estimator, because its current implementation is exceedingly slow in terms of computation and resulted in a memory issue on a Linux machine with Intel(R) Xeon(R) CPU E5-2690 0 @ 2.90GHz. As a result, we subsampled 1\% of the SNPs (by taking every 100th SNP) to obtain results for this estimator, although such subsampling was unnecessary for our efficient implementation of the Constrained Bayes Deconvolution estimator. We have confirmed that our Constrained Bayes Deconvolution estimator gave very similar results for both the full data and the subsampled data. We also ran the R package Kern Smooth to obtain the naive Kernel density estimator that ignores measurement error: as expected, our Constrained Bayes Deconvolution estimator dominated it as well for both the full and subsampled data.

The resulting density estimators are shown in Figure \ref{fig3}. Among the three, our Constrained Bayes Deconvolution method yields a density that has a much sharper peak. This is expected, as in the simulation of Section \ref{sec5.3}, because regular kernel methods, deconvolved or not, cannot handle well this type of very non-standard, but practically important, density.

In addition to the graphical comparison, quantitative comparisons were also made. We compute the estimated probability of the effect size in absolute value being greater than some choices of minimum effect size, displayed in Table \ref{prob_small_eff_GIANT} and Figure \ref{fig8}. As mentioned above, the effect sizes for all SNPs are chosen for our Constrained Bayes Deconvolution and naive Kernel estimators while that of every 100th SNP are selected for the Kernel deconvolution estimator.

\begin{table}[htbp]
\centering
\begin{tabular}{lccccccc}
\hline\hline
  &\multicolumn{7}{c}{Minimum effect size}\\
  \cmidrule(lr){2-8}
Estimator             & 0.002  & 0.0025 & 0.003 & 0.0035 & 0.004 &0.0045 & 0.005\\
\hline
Constrained Bayes     & 0.253  & 0.175 & 0.104 & 0.067 & 0.040 & 0.021 & 0.007  \\*[-.60em]
Kernel                & 0.426  & 0.346 & 0.286 & 0.226 & 0.191 & 0.159 & 0.130  \\*[-.60em]
Naive Kernel          & 0.561  & 0.466 & 0.382 & 0.310 & 0.248 & 0.196 & 0.133  \\
\hline\hline
\end{tabular}
\caption{\baselineskip=12pt Comparison of estimated probability of effect sizes associated with height that the absolute value of effect sizes is greater than the given
minimum effect size under our Constrained Bayes Deconvolution  method (Constrained Bayes), the deconvoluting kernel density estimator (Kernel) and the naive
ordinary kernel density estimator (Naive Kernel) for the GIANT Height effect sizes.}
\label{prob_small_eff_GIANT}
\end{table}

A scientific question in GWAS is to predict the number of significant SNPs for a given sample size, i.e., the number of individuals. Current scientific discoveries are based on the significance of p-values (with a Bonferroni significance level $\alpha = 5 \times 10^{-8}$) for individual SNPs followed by a ``LD clumping" step which selects independent SNPs using their linkage disequilibrium. In recently published GWAS studies of height, \cite{allen2010hundreds}, \cite{wood2014defining}, and \cite{yengo2018meta}, the number of individuals increased from 133K, 253K, to 700K, leading to 180, 697, and 3290 significant discoveries using the described method or more complicated methods regarding the joint SNP effects.

We now briefly discuss the relevance of our density estimation
procedure towards such sample size calculations; additional details
are deferred to Section \ref{sec.S1.3} of \Supp. Suppose
$\wh{\beta} \mid \beta \sim \mbox{N}(\beta, \sigma^2/n)$, where
$\wh{\beta}$ denotes an observed effect size, $\beta$ denotes the
corresponding true effect size with density $f$, and the error
variance $\sigma^2$ is displayed as a constant here for notational
simplicity. A standard approach \citep{chatterjee2013projecting} for
predicting the number of effect sizes achieving genome-wide
significance $\alpha$ at sample size $n$ is provided by the projection
formula, $n \times \mbox{Pr}(\sigma^{-1} \sqrt{n} |\wh{\beta}| > z_{\alpha/2})
= n \int \mbox{pow}_{\sigma, \alpha}(\beta) f(\beta) d\beta$, where
$\mbox{pow}_{\sigma, \alpha}(\beta) = 1 - \Phi(z_{\alpha/2} -
\sqrt{n}\sigma^{-1}\beta) + \Phi(-z_{\alpha/2} -
\sqrt{n}\sigma^{-1}\beta)$.  Here $\Phi(\cdot)$ and $z_{\alpha/2}$
denote the cummulative distribution function and the $(1-\alpha/2)$th
quantile of a standard normal random variable.

We can obtain point and interval estimates for the quantity $\int \mbox{pow}_{\sigma, \alpha}(\beta) f(\beta) d\beta$ from our MCMC output. A Monte Carlo integration is performed to approximate the projection formula using the posterior samples of $\beta$, leading to the desired point prediction. We can further quantify the posterior variability of the predicted number by repeating the calculation on slices dispersed over a MCMC chain. Since scientists are generally interested in the number of independent SNPs that are discovered, we first selected a subset of independent SNPs based on the linkage disequilibrium between the SNPs before estimating the density of $X$ using our procedure. More details about the above procedures can be found in Section \ref{sec.S1.3} of \Supp.

We report in Table \ref{projection_GIANT} the posterior mean of these predicted numbers as our estimator for the expected number of SNPs discovered, together with a $95\%$ credible interval for that number. Although we make an uncommon assumption that none of the effect sizes are exactly zero, our estimates in Table \ref{projection_GIANT} are in the ballpark of the actual numbers from the three cited papers. A clear advantage of using a valid density estimator of true effect sizes in conjunction with the projection formula is that it provides a cheap and simple calculation without carrying out any large-scale experiments. That is, we obtain the density estimator based on the smallest sample size of height study, and quantifies the number of significant SNPs including its uncertainty for larger studies, given no information except their sample sizes. Hence our method can be used to infer the required sample size needed for an expected given number of discoveries.

\begin{table}[htbp]
\centering
\begin{tabular}{lccc}
\hline\hline
  &\multicolumn{3}{c}{Number of individuals}\\
  \cmidrule(lr){2-4}
             & 133K  & 253K & 700K\\
\hline
Exp. Disc.    & 134  & 375 & 2907 \\*[-.60em]
$95\%$ C.I.   & (125, 143)  & (357, 394) & (2790, 3039)\\
\hline\hline
\end{tabular}
\caption{\baselineskip=12pt Estimated value (Exp.Disc.) and a $95\%$
  credible interval ($95\%$ C.I.) for predicting the expected number of SNPs discovered  as the number of individuals varies. We obtain posterior samples of
  the predicted number from the projection formula and posterior samples of effect size distribution.}
\label{projection_GIANT}
\end{table}

\section{Discussion}\label{discuss}

We have considered the case of nonparametric density deconvolution with possibly heteroscedastic measurement errors, where the true densities are subject to shape
constraints, in our case symmetry and unimodality. We are particularly interested in applications where there is a large probability near zero coupled with possibly
heavy tailed distributions. We showed that our method, which we call Constrained Bayes Deconvolution, is nonparametrically consistent for estimating the true target
density in general, and is particularly well-equipped for the mixture problem described immediately above. Computationally, it is linear in the sample size, and hence
highly scalable.

Mixtures of uniforms are known to contain the Normal variance mixture
class \citep{wang2013class} described in Section \ref{sec1}, and have
been utilized in various applications for modeling a symmetric
unimodal density. However, the flexibility of such a model depends
critically on the flexibility of the mixing distribution. Our
carefully designed choice of the Dirichlet process mixture of gammas
for this mixing distribution has large support on the space of
densities on the positive real line, leads to efficient computation,
and is provably consistent. Different approaches, based instead on a
number of mixtures of Normals, include
\cite{stephens2016false}, and a very different approach, based on a
computation in \cite{yang2012conditional}, has been taken by
\cite{zhang2018estimation}, wherein they fit a regression to a large
number of predictors, get the joint regression coefficients, and then
do approximations and linear model calculations to reduce to the
marginal effects, which in this context is our
$X$. \cite{zhu2017bayesian} is a Bayesian approach similar to
\cite{zhang2018estimation}. This particular approach \citep{zhu2017bayesian} seems to be limited to genome-wide association studies based on SNPs, where the
linkage disequilibrium (correlation) between the SNPs is known.

While we are not limited to the effect size context, in that context it might be interesting to replace the idea of a large probability near zero to the case of a point mass
exactly at zero, which has been done in the mixtures of Normals by \cite{stephens2016false} and \cite{zhang2018estimation}. This is possible to do within our
framework and will be reported upon elsewhere. The corresponding results in Table \ref{prob_small_eff_GIANT} are much the same.

\section*{Supplementary Material}

The \Supp\ includes a data analysis of a microarray experiment. The R
code is available from the last author. Code for
  simulations are provided at https://github.com/tamustatsy/Constrained\_Deconvolution/.

\baselineskip=14pt

\section*{Acknowledgments}
Su and Carroll were supported by a grant from the National Cancer Institute (U01- CA057030). Bhattacharya was supported from National Science Foundation grant (NSF DMS 1613156) and a NSF CAREER Award (DMS 1653404). Zhang and Chatterjee were partially funded through a Patient-Centered Outcomes Research Institute (PCORI) Award (ME-1602-34530).
The authors were also supported in part by a grant from the National Human Genome Research Institute (R01-HG010480).
The statements and opinions in this article are solely the responsibility of the authors and do not necessarily represent the views of PCORI, its Board of Governors or Methodology Committee. The authors are grateful to Aurore Delaigle of the University of Melbourne and her collaborators for publishing R package, deconvolve, for 
homoscedastic and heteroscedastic kernel density deconvolution on Github.

\bibliographystyle{biomAbhra}
\bibliography{Bayes_Unimod_Sym_wMe_RJC,CarrollPapers}

\clearpage\pagebreak\newpage
\pagestyle{plain}
\newcommand{\Appendix}{\appendix\def\thesection{Appendix~\Alph{section}}\def\thesubsection{\Alph{section}.\arabic{subsection}}}
\section*{Appendix}
\begin{appendix}
\Appendix\renewcommand{\theequation}{A.\arabic{equation}}
\renewcommand{\thesubsection}{A.\arabic{subsection}}
\renewcommand{\thecondition}{A.\arabic{condition}}
\renewcommand{\theremark}{A.\arabic{remark}}
\renewcommand{\theproposition}{A.\arabic{proposition}}
\renewcommand{\thethm}{A.\arabic{thm}}
\renewcommand{\thelemma}{A.\arabic{lemma}}
\setcounter{equation}{0}
\setcounter{lemma}{0}
\baselineskip=18pt

\subsection{Proof of Theorem \ref{thm_p}} \label{sec:app:6}

Below we provide details to verify (\ref{eq:6}), (\ref{eq:7}) and (\ref{eq:8}) in Section
\ref{thm_p}. 

\cite{bochkina2017adaptive} derive the posterior convergence rate for
Dirichlet location-mixture of Gammas in the no-measurement error
case. We obtain some preliminary results on the layer of $g(\cdot)$
from their work. It is worth pointing out that since the condition on
the Dirichlet process base probability is different from theirs, only
results that are not affected by the type of prior can be inherited
directly in this article. These results can be obtained by Proposition
2.1 and Lemma B.2 in \cite{bochkina2017adaptive}.  Any
$g_0 \in \mathcal{M}$ can be approximated by convoluting a Gamma
kernel and some discrete probability, that is, $K_z \ast P_N$, where
$K_z$ is representing the Gamma kernel with shape and rate parameter
$(z, z/\mu)$ and $P_N$ is a discrete probability
$P_N = \sum_{j = 1}^N p_j \delta_{u_j}$, with
$N \leq N_0 \sqrt{z} (\log z)^{3/2}$, $u_j \in [e_z, E_z]$. The
sequences $\{u_j\}_{j=1}^N$ and $\{p_j\}_{j = 1}^N$ satisfy that
$u_1 = e_z$, $u_N = E_z$, $u_{j + 1} - u_j > z^{-A}$ and
$p_j > z^{-A}$ for some $A > 0$ and with $e_z = z^{-a}$ and
$E_z = z^b$, $a > 1$, $b > 1/\rho_1$, the choice of lower bound on $b$
is larger than that used in \cite{bochkina2017adaptive}, specifically
we require $b > 1/(\rho_1 - 2)$. Define $u_0 = u_1$, $u_{N+1} = u_N$,
then $U_j = [(u_j + u_{j-1})/2, (u_j + u_{j+1})/2]$ covers
$[e_z, E_z]$. Moreover,
$U_0 = \mathbb{R}^+ \setminus \cup_{j = 1}^N U_j$.

Under our Dirichlet location-mixture of Gammas model,
$g(\theta) = K_z \ast G(\theta) = \int g_{z, z/\mu}(\theta) dG(\mu)$,
where the mixing measure $G$ follows $\text{DP}(m, D)$. Define a prior
set
$\mathcal{G}_z = \{G: G(U_i)/p_i \in (1 - 2z^{-A}, 1 - z^{-A}), i = 1,
\ldots, N\}$, while $z \in I_n = (z_n, 2z_n)$. The choice of $z_n$
will be specified later.

In Appendix \ref{sec:app:4} below , we show that on this prior set
$\mathcal{G}_z \times I_n$, the following bounds hold,
\begin{equation}\label{eq:25}
  P_0 \log (p_0/p) \lesssim z_n^{-1} \log(z_n) , \text{ and }
  P_0 \{\log (p_0/p)\}^2 \lesssim  z_n^{-1} \log (z_n) \log(n).
\end{equation}

In Appendix \ref{sec:app:5}, the lower bound for the prior probability
of the prior set $\mathcal{G}_z \times I_n$ is derived, namely that
\begin{equation}
  \label{eq:22}
  \Pi(\mathcal{G}_z \times I_n) \gtrsim \exp\{C \sum_j \log(\alpha_j)\} \gtrsim \exp\{-C z_n^{B+1/2} (\log z_n)^{3/2}\},
\end{equation}
where $B = \max(b a_1, a a_0)$.

Take $z_n = n^{2/(2B + 3)} (\log n)^{-1/(2B + 3)}$, such that
$\epsilon_n^2 = z_n^{-1} \log z_n \asymp n^{-2/(2B + 3)}(\log n)^{(2B
  + 2)/(2B + 3)}$. From (\ref{eq:25}) and (\ref{eq:22}), the prior set
$\mathcal{G}_z \times I_n$ has prior probability bounded below by
$\exp(-C n\epsilon_n^2)$ while on this set
$P_0 \log (p_0/p) \lesssim \epsilon_n^2$,
$P_0 \{\log (p_0/p)\}^2 \lesssim \epsilon_n^2 \log n$. Therefore, the
prior concentration inequality (\ref{eq:8}) holds.

Under the prior in Condition \ref{con4} and the $\epsilon_n$ just
defined, the sieve space on $p(\cdot)$, $\mathcal{F}_n$, in
(\ref{eq:6}) and (\ref{eq:7}) will be defined as follows. Consider a
subspace of $\mathcal{G}$,
\begin{align*}
  Q = Q(\epsilon , J , a , b , \underline{z} , \bar{z}) =
  \{&g(\cdot) = \sum_{j=1}^\infty \pi_j g_{z,z/\mu_j}(\cdot): \sum_{j > J } \pi_j < \epsilon , z \in [\underline{z} ,\bar{z} ],
      \mu_j \in [a , a  + b ] \\
    &\text{ for } j = 1, \ldots, J \}.
\end{align*}
The sieve space of $\mathcal{G}$ is given by
$Q_n = Q(\zeta \epsilon_n, J_n, a_n, b_n, \underline{z}_n,
\bar{z}_n)$. Because of the multi-layer relationship between
$p(\cdot)$, $f(\cdot)$ and $g(\cdot)$ from the definition of
$p(\cdot)$ and $f(\cdot)$, the sieve space on $p(\cdot)$,
$\mathcal{F}_n$, is defined naturally based on $Q_n$. Furthermore, the
entropy and prior mass conditions, (\ref{eq:6}) and (\ref{eq:7}), for
$Q_n$ can be passed along to $\mathcal{F}_n$ due to the fact that the
Hellinger distance between any two functions
$g_1, g_2 \in \mathcal{G}$ is greater than or equal to that between
the corresponding $p_1, p_2$, that is,
$h^2(p_1, p_2) \leq h^2(g_1,g_2)$. It remains to show that $Q_n$
satisfies (\ref{eq:6}) and (\ref{eq:7}).

According to Lemma 4.2 in \cite{bochkina2017adaptive}, (\ref{eq:7})
holds for $Q_n$ if for some positive constant $c$,
\begin{align}\label{eq:3}
  J_n D\{(0,a_n)\} \notag &\lesssim \exp(- c n \epsilon_n^2), \enspace J_n D\{(a_n + b_n, \infty)\} \lesssim \exp(- c n \epsilon_n^2),\\
  1 - \Pi_z([\underline{z}_n, \bar{z}_n]) &\lesssim \exp(- c n \epsilon_n^2), \enspace \{e m J_n^{-1} \log(1/\epsilon_n)\}^{J_n} \lesssim \exp(- c n \epsilon_n^2).
\end{align}

Equation (\ref{eq:6}) holds for $Q_n$ if
\begin{equation}\label{eq:5}
  J_n\{\log \log(b_n/a_n) + \log(\bar{z}_n) + \log(1/\epsilon_n)\} + \log \log(\bar{z}_n/\underline{z}_n) \lesssim n\epsilon_n^2.
\end{equation}

For notational simplicity, let $\eta = 2B + 3$, and set $C > 0$ as a
large enough constant. These conditions are met (details can be found
in Appendix \ref{sec:app:3}) for the following choices of
$J_n = C n^{(2B + 1)/\eta} (\log n)^{-1/\eta}, a_n = C \{n^{(2B +
  1)/\eta} (\log n)^{(2B + 2)/\eta}\}^{-(1/a_0^\prime)}, b_n = C
\{n^{(2B + 1)/\eta} (\log n)^{(2B + 2)/\eta}\}^{(1/a_1^\prime)}$,
$\underline{z}_n = 1 + \exp\{-C n^{(2B + 1)/\eta} (\log n)^{(2B +
  2)/\eta}\}$, $\bar{z}_n = C n^{2(2\beta + 1)/\eta}$
$(\log n)^{2\{(2B + 2)/\eta - \rho_z\}}$.

\subsubsection{Kullback-Leibler Bound}\label{sec:app:4}
One useful result from \cite{bochkina2017adaptive} (in the proof of
their Lemma B.3) is that for any $z$, and
$G \in \mathcal{G}_z = \{G: G(U_i)/p_i \in (1 - 2z^{-A}, 1 - z^{-A}),
i = 1, \ldots, N\}$, it is proved that $h^2(g_0, g) \lesssim z^{-1}$,
where $g = K_z \ast G$. Moreover, it has been shown (in the proof of
their Lemma B.3) that $g(\theta) = (K_z \ast G) (\theta)$ satisfies
\begin{eqnarray}
  \label{eq:17}
  g(\theta) \gtrsim
  \begin{cases}
    z^{-A + 1/2 - M^2/2} & \theta \in [e_z, E_z],\\
    \exp\{2z \log(\theta/e_z) - c \log z\} & \theta < e_z, \\
    \exp(-2z\theta/e_z) & \theta > E_z.
  \end{cases}
\end{eqnarray}

\cite{bochkina2017adaptive} also contains the following lemma (Lemma
C.2 in their paper) which we will make use of to find the tail
probability of the integral with respect to $g$, which is stated as
Lemma \ref{tail_mix_gamma} below.

\begin{lemma}\label{gamma_ineq}
  For all $\delta \in (0,1)$ there exists $c(\delta) > 0$ such that for all $z$ large enough and $u < 1 - \delta$,
  \begin{equation*}
    \frac{z^z \exp{(-z/u)}}{\Gamma(z) u^z} \leq \exp\{-c(\delta)z/u\}.
  \end{equation*}
\end{lemma}

Now we state our Lemma which makes use of Lemma \ref{gamma_ineq} to bound the tail probability of the integral with respect to $g$. The proof is given in Section \ref{sec:app:1}.

\begin{lemma}\label{tail_mix_gamma}
  For all $z$ large enough such that Lemma \ref{gamma_ineq} holds, we have
  \begin{equation*}
    \int_{\theta < 2E_z} g(\theta) d\theta \geq 1 - z^{-1}
    \exp\{-2c(0.5)z\} - z^{-A}.
  \end{equation*}
\end{lemma}

The following inequality, by Lemma 4 of \cite{shen2013adaptive}, can be used to bound the quantities $P_0 \log
(p_0/p)$ and $P_0 \{\log (p_0/p)\}^2$. There exists a $\lambda_0$ such
that for any $\lambda \in (0, \lambda_0)$ and any two densities $p$
and $q$ ($P$ denotes the probability distribution with respect to $p$), 
\begin{eqnarray}
  \label{eq:9} P \log (p/q) &\leq& h^2(p,q) (1 - 2 \log \lambda) + 2 P\{\log (p/q) I(q/p \leq \lambda)\}, \\ \label{eq:10} P \{\log (p/q)\}^2 &\leq& h^2(p,q) \{12 + 2
                                                                                                                                                     (\log \lambda)^2\} + 8 P[\{\log (p/q) \}^2 I(q/p \leq \lambda)].
\end{eqnarray}

We will use $\phi_\sigma$ to denote a Normal density with mean zero and standard deviation $\sigma$. Since
\begin{eqnarray}
  \label{eq:13}
  \int f^{1/2}(u) f_0^{1/2}(u) du \notag&=& \int \int
                                            \phi_\sigma(w -
                                            u)f^{1/2}(u) f_0^{1/2}(u) du dw \\
  \notag&\leq& \int \bigg\{\int \phi_\sigma(w-u) f(u) du \bigg\}^{1/2}
               \bigg\{\int \phi_\sigma(w-u)f_0(u) du \bigg\}^{1/2} dw \\
                                        &=& \int p^{1/2}(w) p_0^{1/2}(w) dw.
\end{eqnarray}
\begin{eqnarray}
  \label{eq:14}
  \int g^{1/2}(\theta) g_0^{1/2}(\theta) d\theta \notag &\leq& \int \int (2\theta)^{-1} I_{(-\theta \leq u < \theta)} g^{1/2}(\theta) g_0^{1/2}(\theta) d\theta du\\
  \notag&\leq&\int \bigg\{\int (2\theta)^{-1} I_{(-\theta \leq u \leq \theta)} g(\theta) d\theta \bigg\}^{1/2} \int \bigg\{ (2\theta)^{-1} I_{(-\theta \leq u \leq \theta)}
               g_0(\theta) d\theta \bigg\}^{1/2} du \\
                                                        &=& \int f^{1/2}(u) f_0^{1/2}(u) du.
\end{eqnarray}
Making use of (\ref{eq:13}) and (\ref{eq:14}), together with the fact that $1 - h^2(p_1,p_2)/2 = \int p_1^{1/2}(x) p_2^{1/2}(x) dx$ holds for any two integrable functions $p_1, p_2$ and the previous result from \cite{bochkina2017adaptive} about $h^2(g, g_0)$, we obtain that
\begin{equation}
  \label{eq:15}
  h^2(p, p_0) \lesssim z^{-1}.
\end{equation}

Suppose $p_0(\cdot)$ has an upper bound $K$. For $|w| < E_z - \delta_N - z^{-A}$, where $\delta_N = u_N - u_{N - 1} > z^{-A}$,
\begin{eqnarray}
  \label{eq:16}
  p(w)/p_0(w) \notag&\geq& K \sigma^{-1} \int_{w -
                           z^{-A}}^{w + z^{-A}} \exp(-(w -
                           u)^2/\sigma^2) \int_{|u|}^\infty \theta^{-1}
                           g(\theta)
                           d\theta du \\
  \notag&\geq& 2 K \sigma^{-1} z^{-A} \exp(-z^{-2A}/\sigma^2) \int_{|w| +
               z^{-A}}^\infty \theta^{-1} g(\theta) d\theta\\
  \notag&\geq& 2 K \sigma^{-1} z^{-A} \exp(-z^{-2A}/\sigma^2)
               \int_{E_z - \delta_N}^\infty \theta^{-1} g(\theta)
               d\theta \\
  \notag&\gtrsim& z^{-A} \int_{E_z - \delta_N}^{E_z} \theta^{-1} g(\theta)
                  d\theta \\
                    &\gtrsim& E_z^{-1} z^{-A + 1/2 - M^2/2} z^{-A} z^{-A} = z^{-3A - b
                              +1/2 - M^2/2}.
\end{eqnarray}
The last inequality is a result of (\ref{eq:17}).

On the other hand, when $|w| > E_z - \delta_N - z^{-A}$, so that when $z$ is large, $w^2 > E_z^2/2$,
\begin{eqnarray}
  \label{eq:18}
  p(w)/p_0(w) \notag&\geq& K \sigma^{-1} \int_{|u| \leq 2 E_z} \exp(-(w -
                           u)^2/\sigma^2) \int_{|u|}^\infty \theta^{-1}
                           g(\theta)
                           d\theta du \\
  \notag&\geq&  K \sigma^{-1} \exp(-18 w^2/\sigma^2)\int_{\theta
               \leq 2E_z} g(\theta) d\theta\\
                    &\geq& K \sigma^{-1} [1 - z^{-1} \exp\{-2c(0.5)z\} -z^{-A}] \exp(-18
                           w^2/\sigma^2).
\end{eqnarray}

  According to (\ref{eq:16}), for $\lambda = K^\prime z^{-3A - b +1/2 - M^2/2}$, if $K^\prime$ is small enough, $\{w: p(w)/p_0(w) \leq \lambda\} \subset \{|w| > E_z
  - \delta_N - z^{-A}\}$. On the latter set, $p_0/p$ is upper bounded as shown in (\ref{eq:18}). Therefore,
  \begin{eqnarray} \label{eq:19}
    P_0[\{\log (p_0/p)\}^2
    I(p/p_0 \leq \lambda)] \lesssim \int_{|w| >
    E_z - \delta_N - z^{-A}} w^4 p_0(w) dw.
  \end{eqnarray}

Our next result, Lemma \ref{tail_p}, is proved in Section \ref{sec:app:1}.
\begin{lemma}\label{tail_p}
 {\rm  Under Condition \ref{con2}, when $t$ is large, $\int_t^\infty w^4 p_0(w) dw \lesssim t^{- \rho_1 + 2}$.}
  \end{lemma}

  Immediately, Lemma \ref{tail_p} leads to an upper bound of (\ref{eq:19}),
  \begin{equation}\label{eq:21}
    P_0[\{\log (p_0/p)\}^2 I(p/p_0 \leq \lambda)] \lesssim
    z^{-b(\rho_1 - 2)} \leq z^{-1},
  \end{equation}
  the last inequality making use of the property of $b$ that $b > 1/(\rho_1 - 2)$.

 Based on (\ref{eq:15}) and (\ref{eq:21}), we can apply (\ref{eq:9}) and (\ref{eq:10}) with the choices of $\lambda$ the same as the one used in (\ref{eq:19}), $p =
 p_0$ and $q = p$, and derive that
  \begin{eqnarray}
    \label{eq:20}
       P_0 \log (p_0/p) \lesssim z^{-1} \log(z); \enspace P_0 \{\log (p_0/p)\}^2 \lesssim z^{-1} \log(z) \log(n).
  \end{eqnarray}

  In summary, (\ref{eq:20}) holds whenever $g \in \mathcal{G}_z$, for any $z$. Hence on the prior set $\mathcal{G}_z \times I_n$, $P_0 \log (p_0/p) \lesssim z_n^{-1}
  \log(z_n)$, $P_0 \{\log (p_0/p)\}^2 \lesssim z_n^{-1} \log(z_n) \log(n)$.

\subsubsection{Prior Probability Bound}\label{sec:app:5}
 Under the new set of priors in Condition \ref{con4}, the prior probability of the prior set $\mathcal{G}_z \times I_n$ has to be modified in the following way. The
 techniques in \cite{bochkina2017adaptive} still apply. The only modification lies in the rate of $\alpha_j = m D(U_j), j = 0, \ldots, N$. Note that for large $u_{j-1}
 \gtrsim E_z$,
  \begin{equation*}
    \alpha_j = m \int_{(u_{j-1} + u_j)/2}^{(u_j + u_{j + 1})/2} d(u) du \gtrsim C \int_{(u_{j-1} + u_j)/2}^{(u_j + u_{j + 1})/2} \exp(- u^{a_1}) du \gtrsim C
    \exp(-E_z^{a_1}) = C \exp(- z^{b a_1}).
  \end{equation*}
For small $0 < u_{j+1} \lesssim e_z$,
  \begin{equation*}
    \alpha_j = m \int_{(u_{j-1} + u_j)/2}^{(u_j + u_{j + 1})/2} d(u) du \gtrsim C \int_{(u_{j-1} + u_j)/2}^{(u_j + u_{j + 1})/2} \exp(- u^{- a_0}) du \gtrsim C
    \exp(-e_z^{-a_0}) = C \exp(-z^{a a_0}).
  \end{equation*}

  Denote $B = \max(b a_1, a a_0)$. For simplicity, we assume without loss of generality that $B = ba_1$. From the above results, $\sum_j (- \log \alpha_j) \lesssim N
  z^B \asymp z^{B + 1/2} (\log z)^{3/2}$. Then we can repeat the lines in the proof of Lemma 4.1 in \cite{bochkina2017adaptive}, so that for $z \in I_n$,
  \begin{equation*}
    \Pi(\mathcal{G}_z) \gtrsim \exp\{C \sum_j \log(\alpha_j)\} \gtrsim \exp\{-C z_n^{B+1/2} (\log z_n)^{3/2}\}.
  \end{equation*}
  On the other hand,
  \begin{equation*} 
    \Pi_z(I_n) \gtrsim \exp\{-C \sqrt{z_n}(\log z_n)^{\rho_z}\}.
  \end{equation*}

\subsubsection{Verification of (\ref{eq:6}) and (\ref{eq:7}) on the Sieve Space $Q_n$}\label{sec:app:3}

In this section, we are going to verify the set of inequalities (\ref{eq:3}) and (\ref{eq:5}) in Appendix \ref{sec:app:6}. Again, our choices of the sieve space parameters
are $J_n = C n^{(2B + 1)/\eta} (\log n)^{-1/\eta}$, $a_n = C \{n^{(2B + 1)/\eta} (\log n)^{(2B + 2)/\eta}\}^{-(1/a_0^\prime)}$, $b_n = C \{n^{(2B + 1)/\eta} (\log
n)^{(2B + 2)/\eta}\}^{(1/a_1^\prime)}$, $\underline{z}_n = 1 + \exp\{-C n^{(2B + 1)/\eta} (\log n)^{(2B + 2)/\eta}\}$, and $\bar{z}_n = C n^{2(2\beta + 1)/\eta}(\log
n)^{2((2B + 2)/\eta - \rho_z)}$.

 Plugging these values together with the condition on the prior,
 \begin{align*}
   J_n D\{(0, a_n)\} &= C n^{(2B + 1)/\eta} (\log n)^{-1/\eta} \int_0^{a_n} d(u) du \\
                        &\lesssim  n^{(2B + 1)/\eta} (\log n)^{-1/\eta}  \exp(-a_n^{-a_0^\prime}) \lesssim \exp(-cn\epsilon_n^2);\\
   J_n D\{(a_n + b_n, \infty)\} &= C n^{(2B + 1)/\eta} (\log n)^{-1/\eta} \int_{a_n + b_n}^{\infty} d(u) du \\
                        &\lesssim  n^{(2B + 1)/\eta} (\log n)^{-1/\eta} \exp(-b_n^{a_1^\prime}) \lesssim \exp(-cn\epsilon_n^2);\\
   \Pi_z\{(1, \underline{z}_n)\} &\lesssim (\underline{z}_n - 1)^{c_0} \lesssim \exp(- c n \epsilon_n^2);\\
   \Pi_z\{[\bar{z}_n, \infty)\} &\lesssim \exp(-c^\prime \sqrt{\bar{z}_n} (\log \bar{z}_n)^{\rho_z}) \lesssim \exp(- c n \epsilon_n^2).
 \end{align*}
 To see that $\{e m J_n^{-1} \log(1/\epsilon_n)\}^{J_n} \lesssim \exp(- c n \epsilon_n^2)$, it is sufficient to show that
 \begin{equation*}
   J_n[\log (J_n) - \log \log(n) + C] \gtrsim c n^{(2B + 1)/\eta} (\log n)^{(2B + 2)/\eta},
 \end{equation*}
 which holds for $J_n = C n^{(2B + 1)/\eta} (\log n)^{-1/\eta}$.

Lastly, we can easily check that the sufficient inequality for (\ref{eq:6}) is valid,
\begin{equation*}
 J_n[\log \log(b_n/a_n) + \log(\bar{z}_n) + \log(1/\epsilon_n)] + \log  \log(\bar{z}_n/\underline{z}_n) \lesssim n\epsilon_n^2.
\end{equation*}

\subsection{Proof of Lemma \ref{inv_ine}}\label{sec:app:2}

Denote $K$ as a symmetric density, whose Fourier transform $\wh{K}$ has support $[-1,1]$. Moreover, $K$ has bounded moments up to order $s$ ($s > 2$). Let
$K_\delta(\cdot) = \delta^{-1} K(\cdot/\delta)$ be its mollifier. Let $g_\delta$ be a function whose Fourier transform $\wh{g}_\delta$ equals
$\wh{K}_\delta/\wh{\phi}_\sigma$, the ratio between the Fourier transform of the kernel $K_\delta$ and that of the Gaussian kernel $\phi_\sigma$.

By the triangular inequality,
\begin{equation}\label{eq:app:8}
  W_2^2(f, f_0) \lesssim W_2^2(f, f*K_\delta) + W_2^2(f_0, f_0*K_\delta) + W_2^2(f*K_\delta ,f_0*K_\delta).
\end{equation}

For the first and second term, based on the property of Wasserstein distance and convolution, the techniques in \cite{nguyen2013convergence} can be used to show
that $W_2^2(f, f*K_\delta) \lesssim \delta^2$, $W_2^2(f_0, f_0*K_\delta) \lesssim \delta^2$.

For the third term in (\ref{eq:app:8}), we first follow the route in Lemma 7 of \cite{gao2016posterior} which makes use of Theorem 6.15 in \cite{villani2008optimal}
stating that the Wasserstein distance $W_k(H_1, H_2)$ is upper bounded by a multiple of the $k$th root of $\int |x|^k d|H_1 - H_2|(x)$,
\begin{equation*} 
  W_2^2(f*K_\delta, f_0*K_\delta) \lesssim \bigg(\int_{|x| \leq M} + \int_{|x| > M}\bigg) |x|^2 |(f-f_0)*K_\delta(x)|dx = T_1 + T_2,
\end{equation*}
say. We will work on $T_1$ and $T_2$ separately.

By the Cauchy Schwartz inequality,
\begin{equation*}  
  T_1 \leq M^{2 + 1/2} \|f*K_\delta - f_0*K_\delta\|_2.
\end{equation*}
Using the arguments in Corollary 2 of \cite{donnet2014posterior},
\begin{eqnarray*}  
  \|f*K_\delta - f_0*K_\delta\|_2 \notag&=& \|(f*\phi_\sigma)*g_\delta - (f_0*\phi_\sigma)*g_\delta\|_2 = \|p*g_\delta - p_0*g_\delta\|_2\\
  &\leq& \|p - p_0\|_1 \|g_\delta\|_2 \lesssim \|p - p_0\|_1 \exp(\sigma^2\delta^{-2}/2).
\end{eqnarray*}

On the other hand,
\begin{align*} 
  T_2 \notag&\leq M^{-(s - 2)} \int_{|x| > M} |x|^s [(f + f_0)*K_\delta(x)] dx \\
      \notag&\lesssim M^{-(s - 2)} \int \int (|x - y|^s + |y|^s) (f + f_0)(x - y) K_\delta(y) dx dy\\
      &\lesssim M^{-(s - 2)} \int |y|^s K_\delta(y) dy + M^{-(s - 2)} \int |x|^s (f + f_0)(x) dx.
\end{align*}
The $s$th moment of $K_\delta$ is finite according to the assumption on $K$, moreover, the $s$th moment of $f_0$ is also finite under the fact that the $s$th moment
of $f_0$ is equivalent to the $s$th moment of $g_0$ and Condition \ref{con2} whenever $s \leq 4$. To make precise what the upper bound for $T_2$ is, it remains to
check the $s$th moment of $f$.

We consider $f(x) = \int I_{(|x| < \theta)} (2\theta)^{-1} g(\theta) d\theta$, with $g$ in the sieve space $Q_n$ in Section \ref{sec3}.
\begin{align}
  \label{eq:app:13}
  \int |x|^s f(x) dx \notag&= \int |x|^s \int I_{(|x| < \theta)} (2\theta)^{-1} g(\theta) d\theta dx = \int \bigg(\int |x|^s I_{(|x| < \theta)} dx\bigg) (2\theta)^{-1} g(\theta)
  d\theta\\
  \notag&\asymp \int \theta^s g(\theta) d\theta = \sum_j \pi_j (z/\mu_j)^{-s} \Gamma(z + s)/\Gamma(z) \lesssim b_n^s,
\end{align}
the last $\lesssim$ is because $\mu_j$ has the upper bound $b_n = C \{n^{(2B + 1)/\eta} (\log n)^{(2B + 2)/\eta}\}^{(1/a_1^\prime)}$.

Plugging the pieces into (\ref{eq:app:8}),
\begin{equation}
  \label{eq:app:14}
   W_2^2(f, f_0) \lesssim \delta^2 + M^{2 + 1/2} \exp(\sigma^2 \delta^{-2}/2) \|p - p_0\|_1 + M^{-(s - 2)} b_n^s.
\end{equation}

The next Lemma is used to select the choice of $M$ in (\ref{eq:app:14}).
\begin{lemma}\label{M_exist}
  As long as $\rho_1$ and $a_1^\prime$ are large enough, there exist some $\nu_1, \nu_2 > 0$ $M = \|p - p_0\|_1^{-2/5 + \nu_1}$ such that $M^{-(s - 2)} b_n^s =
  o_p(\|p - p_0\|_1^{\nu_2})$.
\end{lemma}

From Lemma \ref{M_exist}, we can take $M = \|p - p_0\|_1^{-2/5 + \nu_1}$. In (\ref{eq:app:14}) the optimal value is achieved at $\delta \asymp \{-\log(M^{5/2} \|p -
p_0\|_1)\}^{-1/2} \asymp \{-\log(\|p - p_0\|_1)\}^{-1/2}$. With this choice of $M$ and $\delta$, the second and third term are of order $o(\|p - p_0\|_1^{\nu_1})$ and
$o_p(\|p - p_0\|_1^{\nu_2})$, both are of smaller order than the first term $\delta^2 \asymp \{-\log(\|p - p_0\|_1)\}^{-1}$. Thus we have established that $W_2^2(f, f_0)
\lesssim \{-\log(\|p - p_0\|_1)\}^{-1}$ whenever $g$ is in the sieve space $Q_n$.

\begin{remark}\label{rem2}
  From Condition \ref{con2}, the tail of $g_0$ needs to decrease with a higher order as $\rho_1$ increases.
\end{remark}
\begin{proof}[Proof of Lemma \ref{M_exist}]
  We have shown in Theorem \ref{thm_p} that $\|p - p_0\|_1 =
  O_p(\epsilon_n)$, where $\epsilon_n = n^{-1/\eta} (\log n)^{(2B +
    2)/(2\eta)}$, and $\eta = 2B + 3$. It is sufficient to prove $b_n^{s/(s-2)} = o(\epsilon_n^{-2/5 + \nu})$ for some $\nu > 0$. From the value of $b_n$ and
    $\epsilon_n$, $b_n^{s/(s-2)} = o(\epsilon_n^{-2/5 + \nu})$ holds if $(2 B + 1)s/\{a_1^\prime (s-2)\} < (2/5 - \nu)$ for some $\nu > 0$. The latter is equivalent to $(2
    B + 1)s/\{a_1^\prime (s-2)\} < 2 /5$. Since $B = b a_1$, $b > 1/\rho_1$, after some manipulation it becomes
  \begin{equation*}
  a_1^\prime > \{s/(s - 2)\} (5/2 + 5 a_1/\rho_1)
\end{equation*}
Recall the natural relation $a_1 > a_1^\prime$. A large value for $\rho_1$ and $a_1^\prime$ will guarantee the validity of the above inequality.
\end{proof}

\subsection{Proofs of Lemmas \ref{tail_mix_gamma} and \ref{tail_p}}\label{sec:app:1}

\begin{proof}[Proof of Lemma \ref{tail_mix_gamma}]
   Recall that $g_{z,\mu}$ denotes a Gamma density with shape $z$ and rate $z/\mu$.
      \begin{eqnarray} \label{eq:app:1}
      \int_{\theta < 2E_z} g(\theta) d\theta \notag&=& 1 - \int_{\theta >
                                                       2E_z} g(\theta) d\theta \\
      \notag&=& 1 - \int_{\theta > 2E_z}
      \int_{\mu < E_z} g_{z,\mu}(\theta)dG(\mu) d\theta
      - \int_{\theta > 2E_z}
      \int_{\mu > E_z} g_{z,\mu}(\theta)dG(\mu) d\theta
        \\
      &=& 1 - \mbf{I} - \mbf{II}, \text{ namely.}
      \end{eqnarray}
    Apply Lemma \ref{gamma_ineq} to $g_{z,\mu}(\theta)$ with
    $\theta > 2E_z$, $\mu < E_z$ such that $\delta = 1/2$,
    \begin{eqnarray}
      \label{eq:app:2}
       \mbf{I} \notag&\leq& \int_{\theta > 2E_z} \theta^{-1} \exp\{- c(0.5) z
                       \theta/E_z\} d\theta \\
      &\lesssim& E_z^{-1} \int_{\theta > 2E_z} \exp\{- c(0.5) z
       \theta/E_z\} d\theta = z^{-1} \exp\{-2c(0.5)z\}.
   \end{eqnarray}
   On the other hand, for any $G \in \mathcal{G}_z$, $G(\mu >
   E_z) \leq z^{-A}$, hence
   \begin{eqnarray}
     \label{eq:app:3}
     \mbf{II} \leq \int_{\theta > 2E_z} g_{z, E_z}(\theta) d\theta
     \int_{\mu > E_z} dG(\mu) \lesssim z^{-A} .
   \end{eqnarray}
   Combining (\ref{eq:app:1}), (\ref{eq:app:2}) and (\ref{eq:app:3}), the desired result is proved.
\end{proof}

\begin{proof}[Proof of Lemma \ref{tail_p}]

  Throughout the proof, we assume that $t$ is any large number.

  Since $W^4 \leq C (X^4 + U^4)$, $P_0(W > t) \leq P_0(X > t/2) + P_0(U > t/2)$,
  \begin{equation}\label{eq:app:4}
    P_0\{W^4 I_{(|W| > t)}\} \lesssim P_0\{X^4 I_{(|X| > t/2)}\} + P_0\{U^4 I_{(|X| > t/2)}\} + P_0\{X^4 I_{(|U| > t/2)}\} + P_0\{U^4 I_{(|U| > t/2)}\}.
  \end{equation}
 Under Condition \ref{tail_p}, it can be easily shown that $P_0(X^4) < \infty$. Moreover, the fourth moment of Normal distribution exists, therefore, $P_0(U^4) <
 \infty$. It follows that the second and third term in (\ref{eq:app:4}) are upper bounded by the first and fourth term correspondingly, thus
  \begin{equation}
    \label{eq:app:5}
    P_0\{W^4 I_{(|W| > t)}\} \lesssim P_0\{X^4 I_{(|X| > t/2)}\} + P_0\{U^4 I_{(|U| > t/2)}\}.
  \end{equation}

Since $U$ follows a Normal distribution which has exponential tail, $P_0\{U^4 I_{(|U| > t/2)}\} \lesssim t^{-\rho_1 + 2}$. For the proof of the Lemma, it remains to
show the upper bound of the first term on the right hand side of (\ref{eq:app:5}).
  \begin{align*}
    P_0\{X^4 I_{(|X| > t/2)}\} &= P_0[P_0\{X^4 I_{(|X| > t/2)}|\theta\}] = \int (2\theta)^{-1} \int x^4 I_{\{(|x| > t/2) \cap (|x| < \theta)\}} dx g_0(\theta) d \theta  \\
    &\lesssim \int \theta^{-1} I_{(|\theta| > t/2)} \{\theta^5 - (t/2)^5\} g_0(\theta) d \theta \\
    &\leq \int \theta^4 I_{(|\theta| > t/2)} g_0(\theta) d \theta \lesssim (1 + t/2)^{-\rho_1 + 2} \lesssim t^{-\rho_1 + 2}.
  \end{align*}
  The second but last inequality is because of Condition \ref{con2}. This concludes the proof of Lemma \ref{tail_p}.
\end{proof}

\subsection{Major differences in proofs when the error is Laplace }\label{sec:app:7}

We walk through the steps in Section \ref{sec:app:6} and
\ref{sec:app:2} to prove Theorem \ref{thm_p} and Lemma
\ref{inv_ine_lap} correspondingly. Theorem \ref{thm_f} is again a
corollary of the two. Let us denote
$\psi_\sigma = (2 \sigma)^{-1} \exp(-|x|/\sigma)$ as the density of
Laplace distribution with location zero and scale
parameter $\sigma$.

Theorem 1 can be shown by modifying Section \ref{sec:app:6}. We can directly show that only when
deriving the KL type upper bounds in \eqref{eq:25} the
error distribution might play a role. However, it turns out \eqref{eq:25} is not
changing based on the details below.

The lines in \eqref{eq:13} and \eqref{eq:14} go through for any density, in
particular $\psi_\sigma$. So from \eqref{eq:15} it remains true that $h^2(p, p_0) \lesssim
z^{-1}$.

Also, the lower bound \eqref{eq:16} for $p(w)/p_0(w)$ on $|w| < E_z -
\delta_N - z^{-A}$ stays the same, while the lower bound \eqref{eq:18} for
$p(w)/p_0(w)$ on $|w| > E_z - \delta_N - z^{-A}$ changes slightly to $K \sigma^{-1} [1 - z^{-1} \exp\{-2c(0.5)z\} -z^{-A}] \exp(-4
           |w|/\sigma)$. These bounds would yield the upper
           bounds for the KL-type divergence, for $\lambda =
           K^{\prime} z^{-3A - b + 1/2 - M^2/2}$ ($K^{\prime}$ small enough),
\begin{eqnarray*}
P_0\{\log (p_0/p) I(p/p_0 \leq \lambda)\} \lesssim \int_{|w| > E_z - \delta_N - z^{-A}} |w| p_0(w) dw, \\ P_0[\{\log (p_0/p) \}^2 I(p/p_0 \leq \lambda)] \lesssim \int_{|w|
> E_z - \delta_N - z^{-A}} w^2 p_0(w) dw.
\end{eqnarray*}
Under condition \ref{con3} we can show that, along the same lines of proofs for Lemma
\ref{tail_p}, both terms on the right hand side above are bounded
by $z^{-1}$. Hence $P_0\{\log (p_0/p) I(p/p_0 \leq \lambda)\} \lesssim
z^{-1}, P_0[\{\log (p_0/p)\}^2 I(p/p_0 \leq \lambda)] \lesssim z^{-1}$.
Thus \eqref{eq:25} concludes.

Lemma \ref{inv_ine_lap} can be shown by modifying Section \ref{sec:app:2}. We revise
the definition of $g_\delta$ whose Fourier transform $\wh{g}_\delta$
equals $\wh{K}_\delta/\wh{\psi}_\sigma$, the ratio between the Fourier
transform of the kernel $K_\delta$ and that of the Laplace density
$\psi_\sigma$.

As shown in Section \ref{sec:app:2}, the upper bound (up to constant) for the term $\|f * K_\delta
- f_0 * K_\delta\|_2$ is $\|p - p_0\|_1
\|g_\delta\|_2$. The $L_2$ norm of $\|g_\delta\|_2$ is the same as the
$L_2$ norm its Fourier transform, which is bounded by $(1 + \sigma^2 \delta^{-2}) \delta^{-1/2} \asymp \delta^{-5/2}$.

So \eqref{eq:app:14} (the other two terms are not affected by
distribution of error) is modified to
\begin{equation*}
   W_2^2(f, f_0) \lesssim \delta^2 + M^{2 + 1/2} \delta^{-5/2} \|p - p_0\|_1 + M^{-(s - 2)} b_n^s.
\end{equation*}
Lemma \ref{M_exist} still holds for the same choice of $M$, that is,
there exists $\nu_1, \nu_2 > 0$ such that $M = \|p -
p_0\|_1^{-2/5 + \nu_1}$ and $M^{-(s - 2)}
b_n^s = o_p(\|p - p_0\|_1^{\nu_2})$ given that $\rho_1$ and
$a_1^\prime$ are large enough. Then we
can show that the right hand side of the above is $O(\|p -
p_0\|_1^\nu)$. However, the value of $\nu$ is determined by the
interplay of $\rho_1$ and $a_1^\prime$ and does not have a simple form
so we omit writing it out.
\end{appendix}

\clearpage\pagebreak\newpage
\pagestyle{empty}

\begin{table}[htbp]
  \centering
\begin{tabular}  {llcccc}
\hline\hline
&& \multicolumn{2}{c}{Normal} & \multicolumn{2}{c}{Laplace}
  \\*[-.20em]
   \cmidrule{3-4}  \cmidrule{5-6}
              && Constrained & & Constrained\\*[-.60em]
    \( n \) & & Bayes & Kernel & Bayes & Kernel\\
    \hline
    1000 & IAE & 0.107 (0.016) & 0.349 (0.105) & 0.104 (0.012) & 0.185 (0.046)  \\*[-.60em]
    & ISE & 0.072 (0.019) & 0.148 (0.043) & 0.067 (0.012) & 0.079 (0.022)       \\*[-.60em]
     & $W_2$ & 0.165 (0.050) & 0.589 (0.209) & 0.185 (0.054) & 0.285 (0.068)     \\
    5000 & IAE & 0.091 (0.014) & 0.277 (0.047) & 0.081 (0.013) & 0.127 (0.028) \\*[-.60em]
    & ISE & 0.072 (0.014) & 0.120 (0.022) & 0.062 (0.013) & 0.055 (0.015)      \\*[-.60em]
    & $W_2$ & 0.073 (0.017) & 0.445 (0.122) & 0.076 (0.019) & 0.185 (0.038)     \\
\hline  \hline
\end{tabular}
 \caption{\baselineskip=12pt  Comparison of our Constrained Bayes
   Deconvolution method (Constrained Bayes) and the deconvoluting
   kernel density estimator (Kernel). This is in the case when the
   target density is a t-density with 5 degrees of freedom and the measurement errors
   are from Normal or Laplace distribution with homoscedastic
   variance. The sample size is $n$, IAE is integrated absolute error, and ISE is integrated squared error. $W_2$ denotes the Wasserstein distance of order $2$.
   Numbers in parentheses are standard errors. Sample sizes greater than $5000$ yield similar results.
}
  \label{tab2}
\end{table}

\begin{table}[htbp]
  \centering
\begin{tabular}{llcccc}
\hline\hline
&& \multicolumn{2}{c}{Normal} & \multicolumn{2}{c}{Laplace}
  \\*[-.20em]
   \cmidrule{3-4}  \cmidrule{5-6}
              && Constrained & & Constrained\\*[-.60em]
    \( n \) & & Bayes & Kernel & Bayes & Kernel\\
    \hline
  1000 & IAE & 0.086 (0.010) & 0.394 (0.045) & 0.089 (0.010) & 0.220 (0.046) \\*[-.60em]
    & ISE & 0.053 (0.007) & 0.183 (0.023) & 0.052 (0.006) & 0.098 (0.024)  \\*[-.60em]
    & $W_2$ & 0.140 (0.042) & 0.456 (0.060) & 0.154 (0.039) & 0.324 (0.070) \\
    5000 & IAE & 0.057 (0.007) & 0.389 (0.022) & 0.062 (0.007) & 0.189
  (0.056)\\*[-.60em]
    & ISE & 0.035 (0.005) & 0.181 (0.012) & 0.038 (0.004) & 0.085 (0.027)\\*[-.60em]
    & $W_2$ & 0.076 (0.014) & 0.430 (0.032) & 0.083 (0.016) & 0.279 (0.052)\\
    \hline\hline
    \end{tabular}
\caption{\baselineskip=12pt  Comparison of our Constrained Bayes
  Deconvolution method (Constrained Bayes) and the deconvoluting
  kernel density estimator (Kernel). This is in the case when the target density is
   a t-density with 5 degrees of freedom and the measurement errors
   are from Normal or Laplace distribution with heteroscedastic variance. The sample size is $n$,
  IAE is integrated absolute error, and ISE is integrated
  squared error. $W_2$ denotes the Wasserstein
  distance of order $2$. Numbers in parentheses are standard errors. Sample sizes greater than $5000$ yield similar results.
}
  \label{tab4}
\end{table}

\begin{table}[htbp]
  \centering
\begin{tabular}{llcccc}
\hline\hline
&& \multicolumn{2}{c}{Normal} & \multicolumn{2}{c}{Laplace}
  \\*[-.20em]
   \cmidrule{3-4}  \cmidrule{5-6}
              && Constrained & & Constrained\\*[-.60em]
    \( n \) & & Bayes & Kernel & Bayes & Kernel\\
    \hline
    1000 & IAE & 0.326 (0.054) & 0.720 (0.115) & 0.258 (0.044) &
                                                                  0.393 (0.078)\\*[-.60em]
    & ISE & 0.390 (0.065) & 0.572 (0.083) & 0.307 (0.054) & 0.309 (0.075) \\*[-.60em]
    & $W_2$ & 0.109 (0.026) & 0.263 (0.075) &0.086 (0.020) & 0.136
                                                             (0.029)  \\*[-.60em]
    & Exceedance & 0.068 (0.021) & 0.178 (0.048) & 0.049 (0.016) &
                                                                   0.047 (0.025)\\
    5000 & IAE & 0.188 (0.033) & 0.656 (0.059) & 0.139 (0.018) &                                                              0.280 (0.049)\\*[-.60em]
    & ISE & 0.219 (0.038) & 0.530 (0.046) & 0.167 (0.019) & 0.217 (0.049)\\*[-.60em]
    & $W_2$ & 0.057 (0.011) & 0.222 (0.032) & 0.041 (0.011) & 0.087 (0.018) \\*[-.60em]
    & Exceedance & 0.026 (0.010) & 0.146 (0.019) & 0.014 (0.008) &
                                                                   0.023 (0.014)\\
    \hline\hline
  \end{tabular}
  \caption{\baselineskip=12pt
   Comparison of our Constrained Bayes Deconvolution method (Constrained Bayes),
   the deconvoluting kernel density estimator (Kernel). This is in the
   case when the target density is a
   mixture of t-density with 5 degrees of freedom and a Normal density
   with standard deviation $0.2$, and when the measurement errors are
   from Normal or Laplace distribution with homoscedastic variance. The sample size is $n$, IAE is integrated
   absolute error, ISE is integrated squared error and
   Exceedance is the absolute difference between the exceedance
   probability under the estimated and true densities. $W_2$ denotes the Wasserstein distance of order $2$. Numbers in parentheses are standard errors. Sample sizes
   greater than $5000$ yield similar results.}
  \label{tab6}
\end{table}

\begin{table}[htbp]
  \centering
\begin{tabular}{llcccc}
\hline\hline
&& \multicolumn{2}{c}{Normal} & \multicolumn{2}{c}{Laplace}
  \\*[-.20em]
   \cmidrule{3-4}  \cmidrule{5-6}
              && Constrained & & Constrained\\*[-.60em]
    \( n \) & & Bayes & Kernel & Bayes & Kernel\\
    \hline
    1000 & IAE & 0.452 (0.067) & 0.848 (0.053) & 0.359 (0.058) & 0.473
  (0.067)\\*[-.60em]
    & ISE & 0.532 (0.073) & 0.656 (0.033) & 0.422 (0.066) & 0.374 (0.084)\\*[-.60em]
    & $W_2$ & 0.176 (0.036) & 0.360 (0.085) & 0.147 (0.030) & 0.183 (0.051)\\*[-.60em]
    & Exceedance & 0.121 (0.027) & 0.259 (0.040) & 0.092 (0.021) &
                                                                   0.071 (0.048)\\
    5000 & IAE & 0.276 (0.047) & 0.820 (0.024) & 0.183 (0.028) & 0.357
  (0.091)\\*[-.60em]
    & ISE & 0.321 (0.057) & 0.639 (0.014) & 0.211 (0.029) & 0.279 (0.080)\\*[-.60em]
    & $W_2$ & 0.100 (0.017) & 0.309 (0.053) & 0.083 (0.017) & 0.140
                                                              (0.058)\\*[-.60em]
    & Exceedance & 0.065 (0.013) & 0.231 (0.018) & 0.048 (0.010) &
                                                                   0.049 (0.047)\\
    \hline\hline
  \end{tabular}
\caption{\baselineskip=12pt
   Comparison of our Constrained Bayes Deconvolution method (Constrained Bayes),
   the deconvoluting kernel density estimator (Kernel). This is in the
   case when the target density is a mixture of t-density with 5
   degrees of freedom and a Normal density with standard deviation
   $0.2$, and when the measurement errors are
   from Normal or Laplace distribution with heteroscedastic variance. The sample size is $n$, IAE is  integrated absolute error, ISE is  integrated squared error and
   Exceedance is the absolute difference between the exceedance probability under the estimated and true densities. $W_2$ denotes the Wasserstein distance of order
   $2$. Numbers in parentheses are standard errors. Sample sizes greater than $5000$ yield similar results.}
    \label{tab7}
  \end{table}

\begin{figure}[!ht]
\centering
\includegraphics[height=5in, width=5.5in, trim=1cm 1cm 1cm 1.5cm, clip=true]{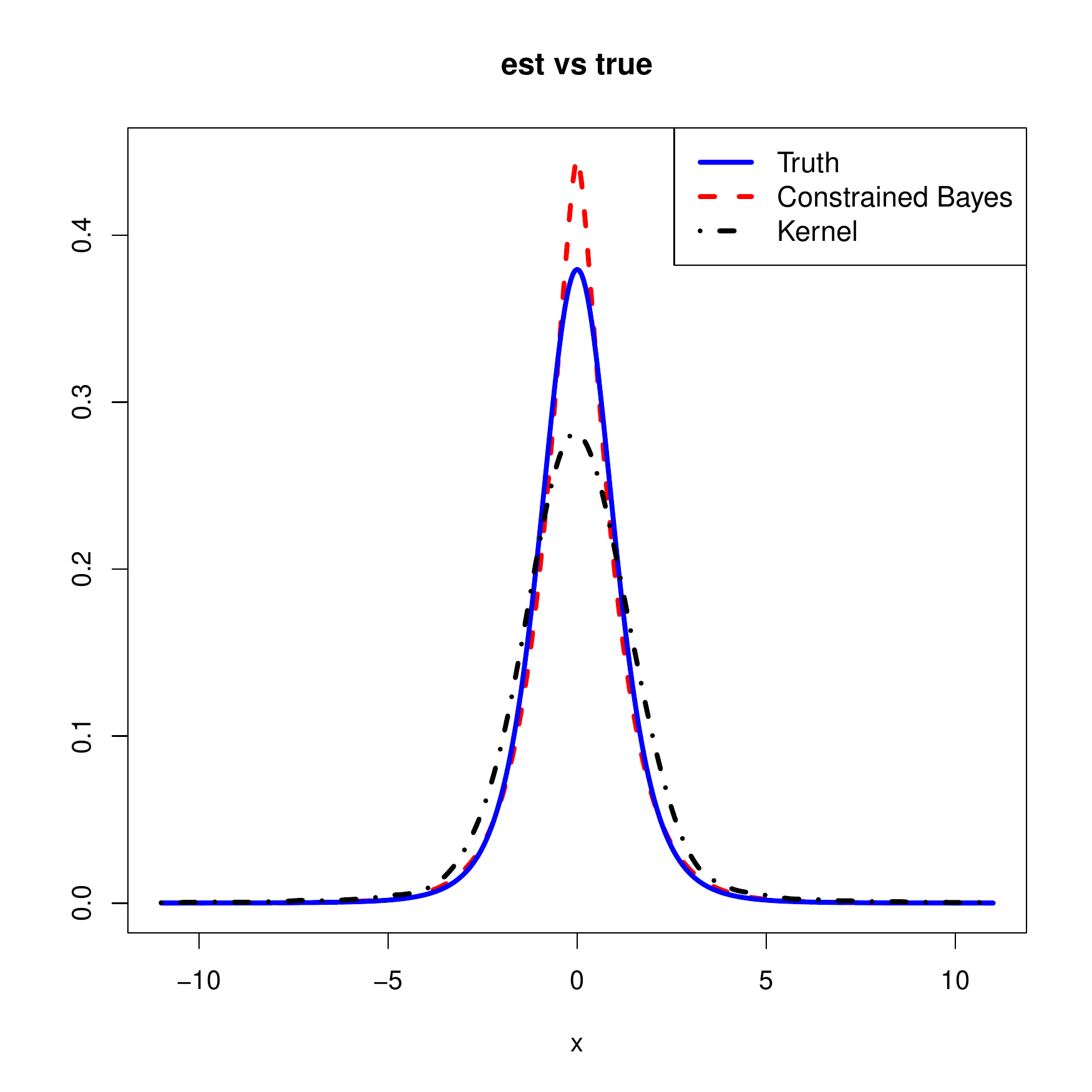}
\vskip -10pt
\caption{\baselineskip=12pt Mean density estimates for the
  homoscedastic Normal measurement error simulation of Section
  \ref{sec5.2} for sample size $n=5000$. Solid blue line is the truth (Truth, a t--density with 5 degrees of freedom), the dashed red line is our Constrained Bayes
  Deconvolution method (Constrained Bayes) and the dash-dotted black line is the deconvoluting kernel density estimator (Kernel).}
\label{sim_fig1}
\end{figure}

\begin{figure}[!ht]
\centering
\includegraphics[height=5in, width=5.5in, trim=1cm 1cm 1cm 1.5cm, clip=true]{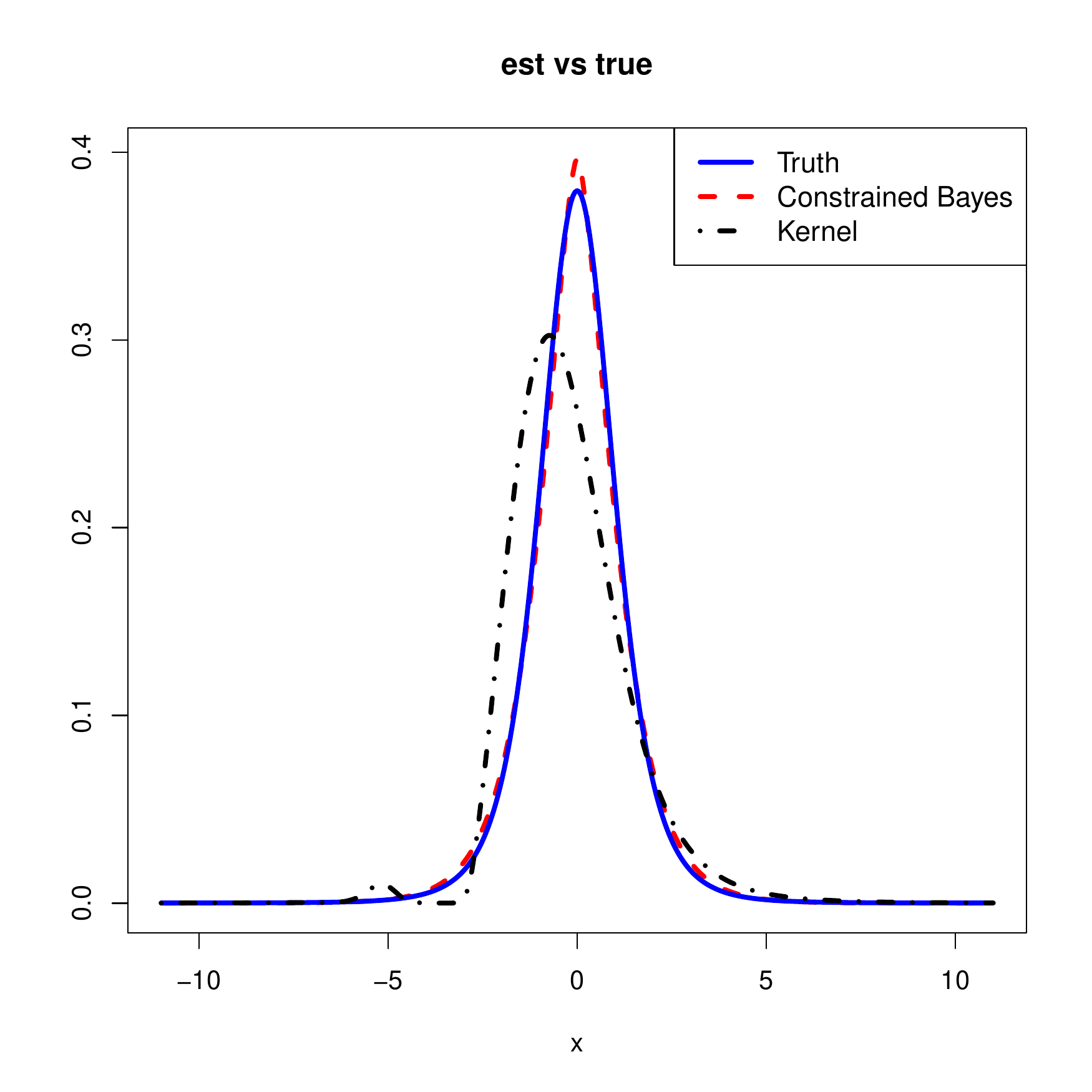}
\vskip -10pt
\caption{\baselineskip=12pt Mean density estimates for the
  heteroscedastic Normal measurement error simulation of Section
  \ref{sec5.2} for sample size $n=5000$. Solid blue line is the truth (Truth, a t--density with 5 degrees of freedom), the dashed red line is our Constrained Bayes
  Deconvolution method (Constrained Bayes) and the dash-dotted black line is the deconvoluting kernel density estimator (Kernel).}
\label{sim_fig5}
\end{figure}

\begin{figure}[!ht]
\centering
\includegraphics[height=5in, width=5.5in, trim=1cm 1cm 1cm 1.5cm, clip=true]{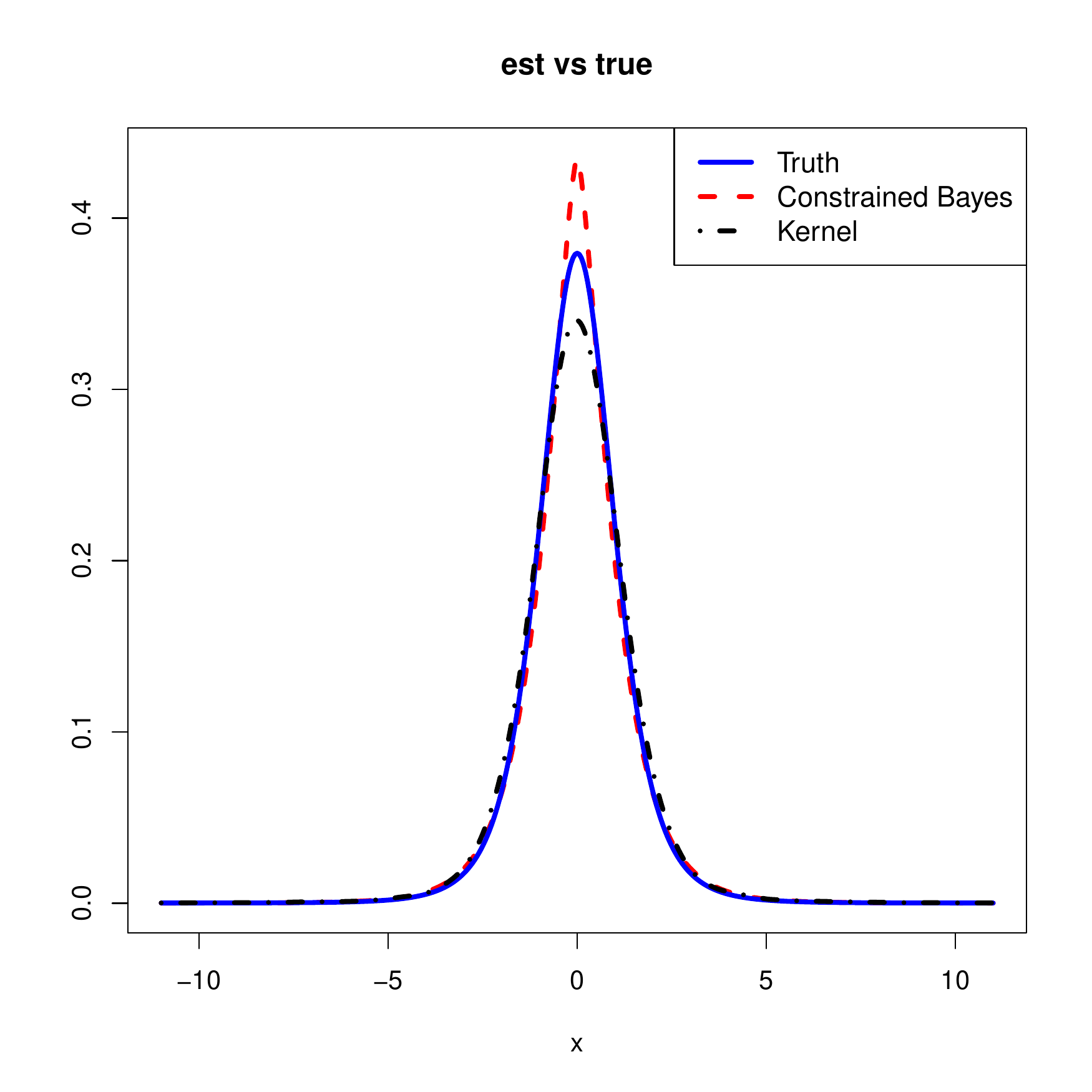}
\vskip -10pt
\caption{\baselineskip=12pt Mean density estimates for the
  homoscedastic Laplace measurement error simulation of Section
  \ref{sec5.2} for sample size $n=5000$. Solid blue line is the truth (Truth, a t--density with 5 degrees of freedom), the dashed red line is our Constrained Bayes
  Deconvolution method (Constrained Bayes) and the dash-dotted black line is the deconvoluting kernel density estimator (Kernel).}
\label{sim_fig2}
\end{figure}

\begin{figure}[!ht]
\centering
\includegraphics[height=5in, width=5.5in, trim=1cm 1cm 1cm 1.5cm, clip=true]{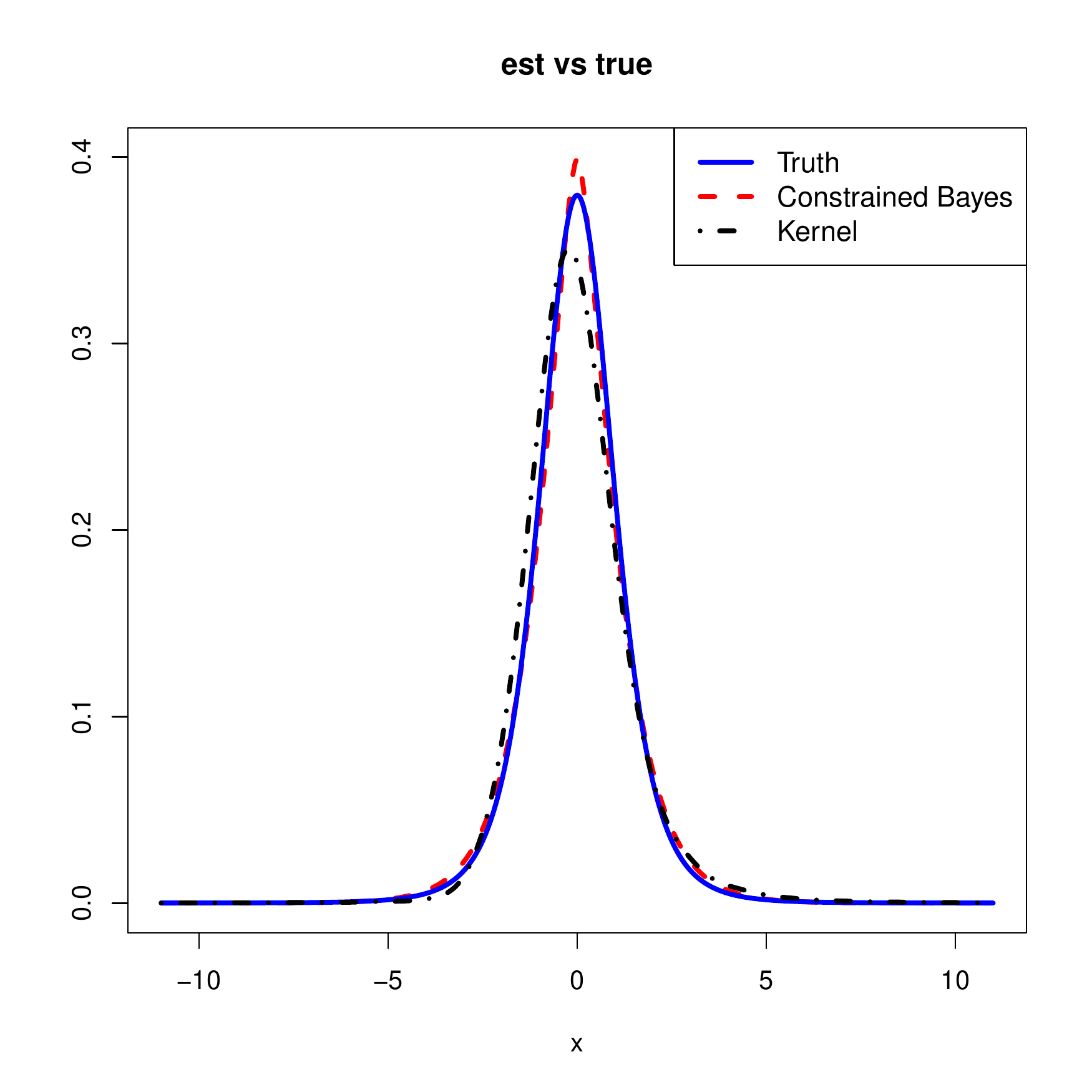}
\vskip -10pt
\caption{\baselineskip=12pt Mean density estimates for the
  heteroscedastic Laplace measurement error simulation of Section
  \ref{sec5.2} for sample size $n=5000$. Solid blue line is the truth (Truth, a t--density with 5 degrees of freedom), the dashed red line is our Constrained Bayes
  Deconvolution method (Constrained Bayes) and the dash-dotted black line is the deconvoluting kernel density estimator (Kernel).}
\label{sim_fig4}
\end{figure}

\begin{figure}[!ht]
\centering
\includegraphics[height=5in, width=5.5in, trim=1cm 1cm 1cm 1.5cm, clip=true]{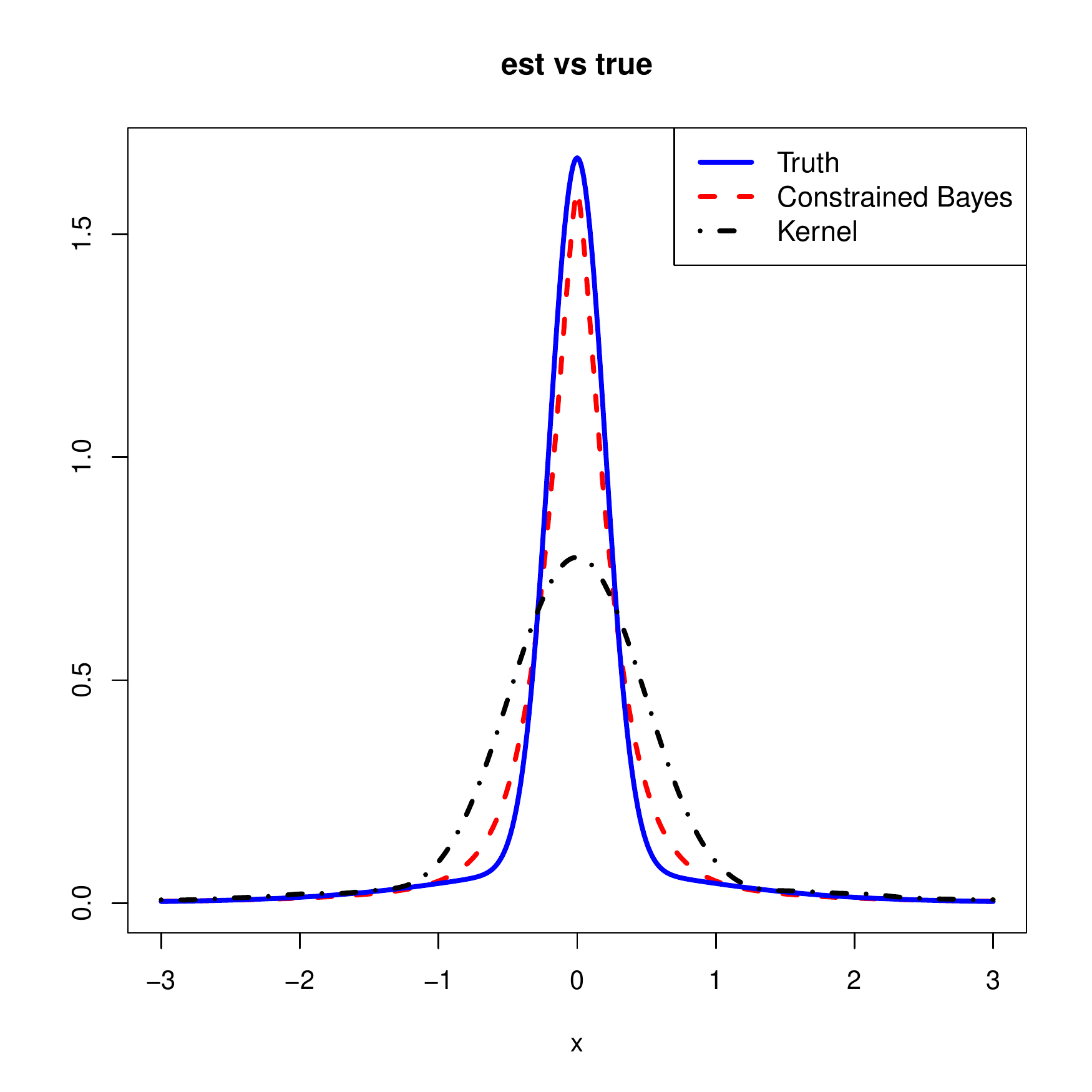}
\vskip -10pt
\caption{\baselineskip=12pt Mean density estimates for the
  homoscedastic Normal measurement error simulation of Section \ref{sec5.3} for sample size $n=5000$. Solid blue line is the truth (Truth, a mixture of a t--density
  with 5 degrees of freedom and a Normal density with standard deviation $0.2$), the dashed red line is our Constrained Bayes Deconvolution method (Constrained
  Bayes) and the dash-dotted black line is the deconvoluting kernel density estimator (Kernel).}
\label{sim_fig6}
\end{figure}

\begin{figure}[!ht]
\centering
\includegraphics[height=5in, width=5.5in, trim=1cm 1cm 1cm 1.5cm, clip=true]{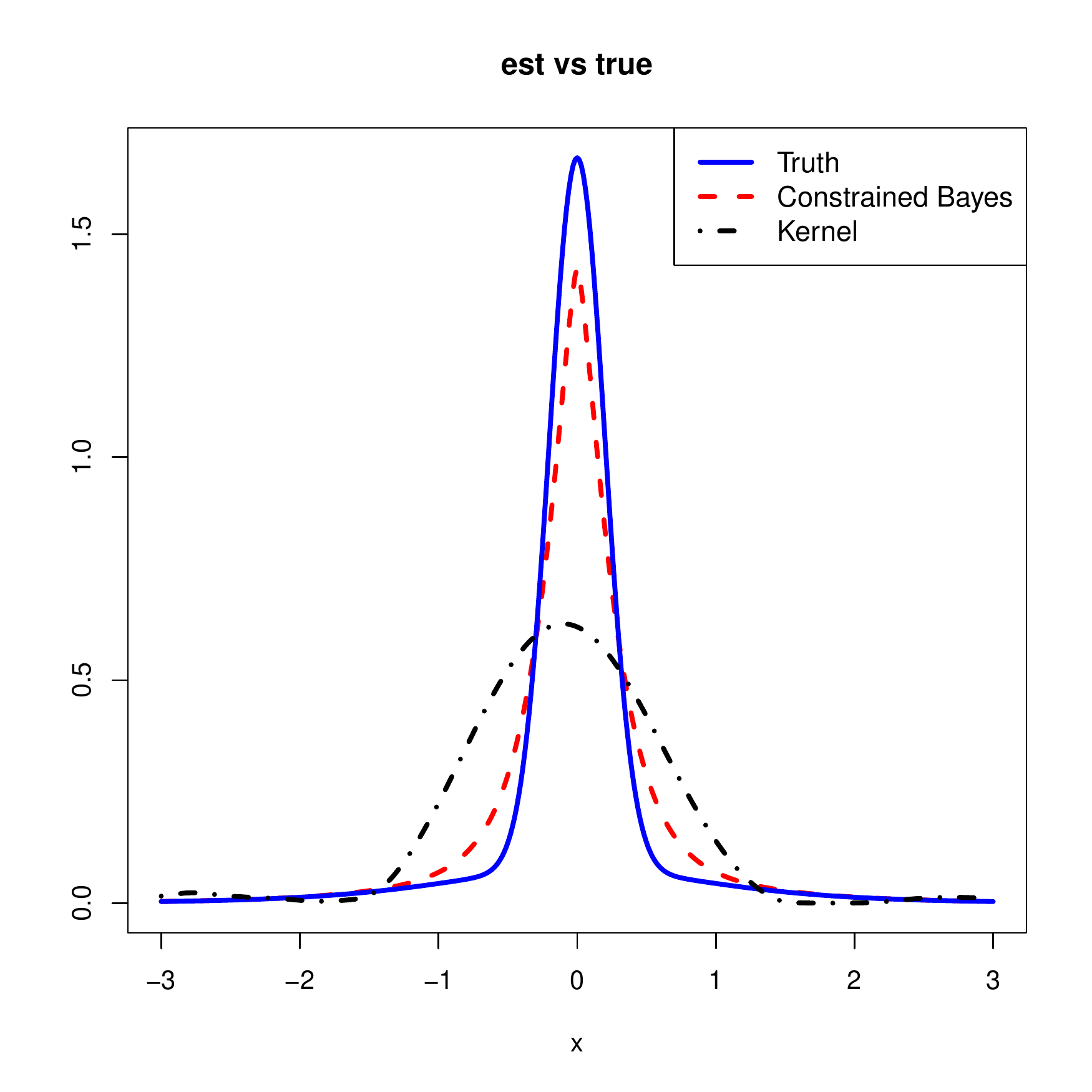}
\vskip -10pt
\caption{\baselineskip=12pt Mean density estimates for the
  heteroscedastic Normal measurement error simulation of Section \ref{sec5.3} for sample size $n=5000$. Solid blue line is the truth (Truth, a mixture of a t--density
  with 5 degrees of freedom and a Normal density with standard deviation $0.2$), the dashed red line is our Constrained Bayes Deconvolution method (Constrained
  Bayes) and the dash-dotted black line is the deconvoluting kernel density estimator (Kernel).}
\label{sim_fig3}
\end{figure}

\begin{figure}[!ht]
\centering
\includegraphics[height=5in, width=5.5in, trim=1cm 1cm 1cm 1.5cm, clip=true]{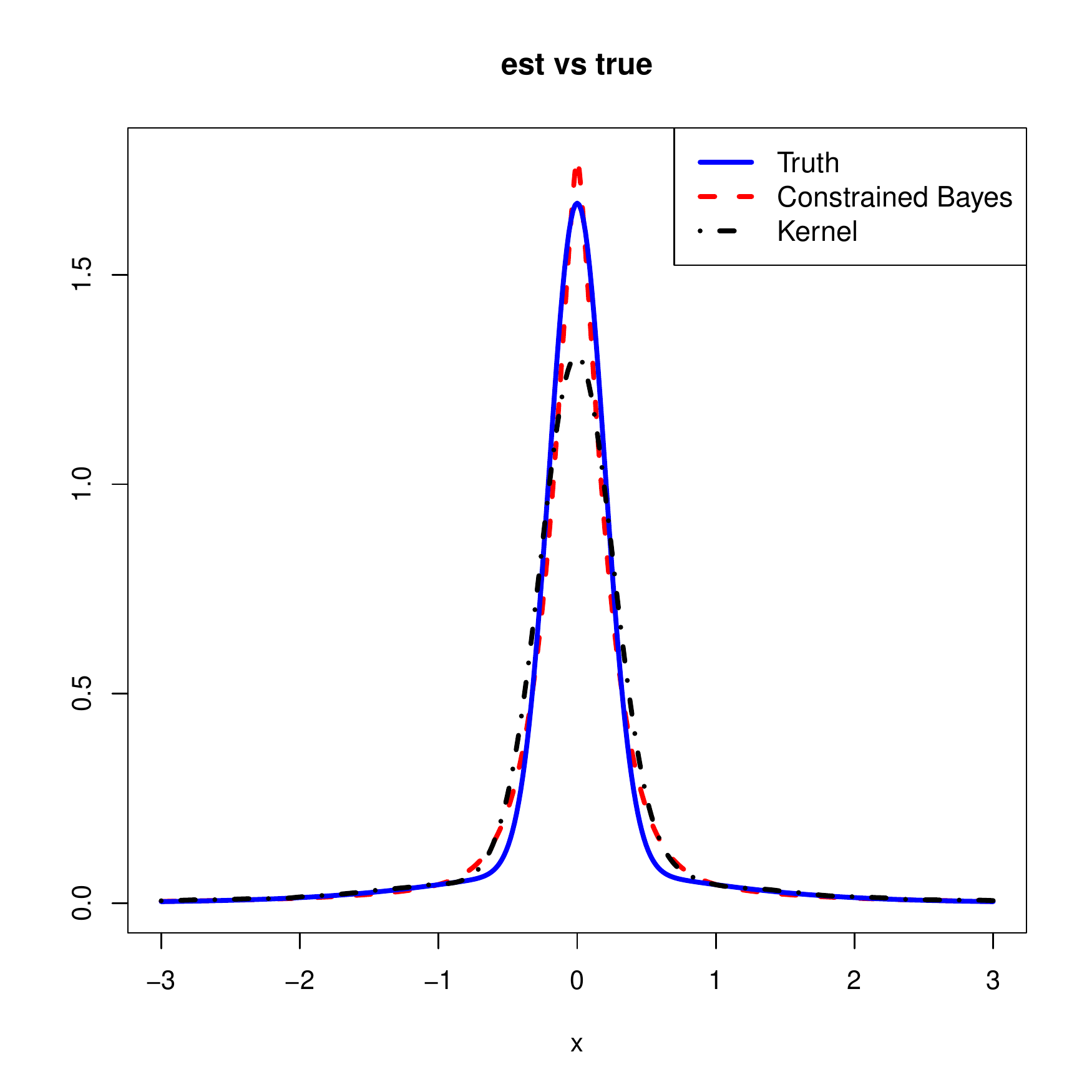}
\vskip -10pt
\caption{\baselineskip=12pt Mean density estimates for the
  homoscedastic Laplace measurement error simulation of Section \ref{sec5.3} for sample size $n=5000$. Solid blue line is the truth (Truth, a mixture of a t--density
  with 5 degrees of freedom and a Normal density with standard deviation $0.2$), the dashed red line is our Constrained Bayes Deconvolution method (Constrained
  Bayes) and the dash-dotted black line is the deconvoluting kernel density estimator (Kernel).}
\label{sim_fig7}
\end{figure}

\begin{figure}[!ht]
\centering
\includegraphics[height=5in, width=5.5in, trim=1cm 1cm 1cm 1.5cm, clip=true]{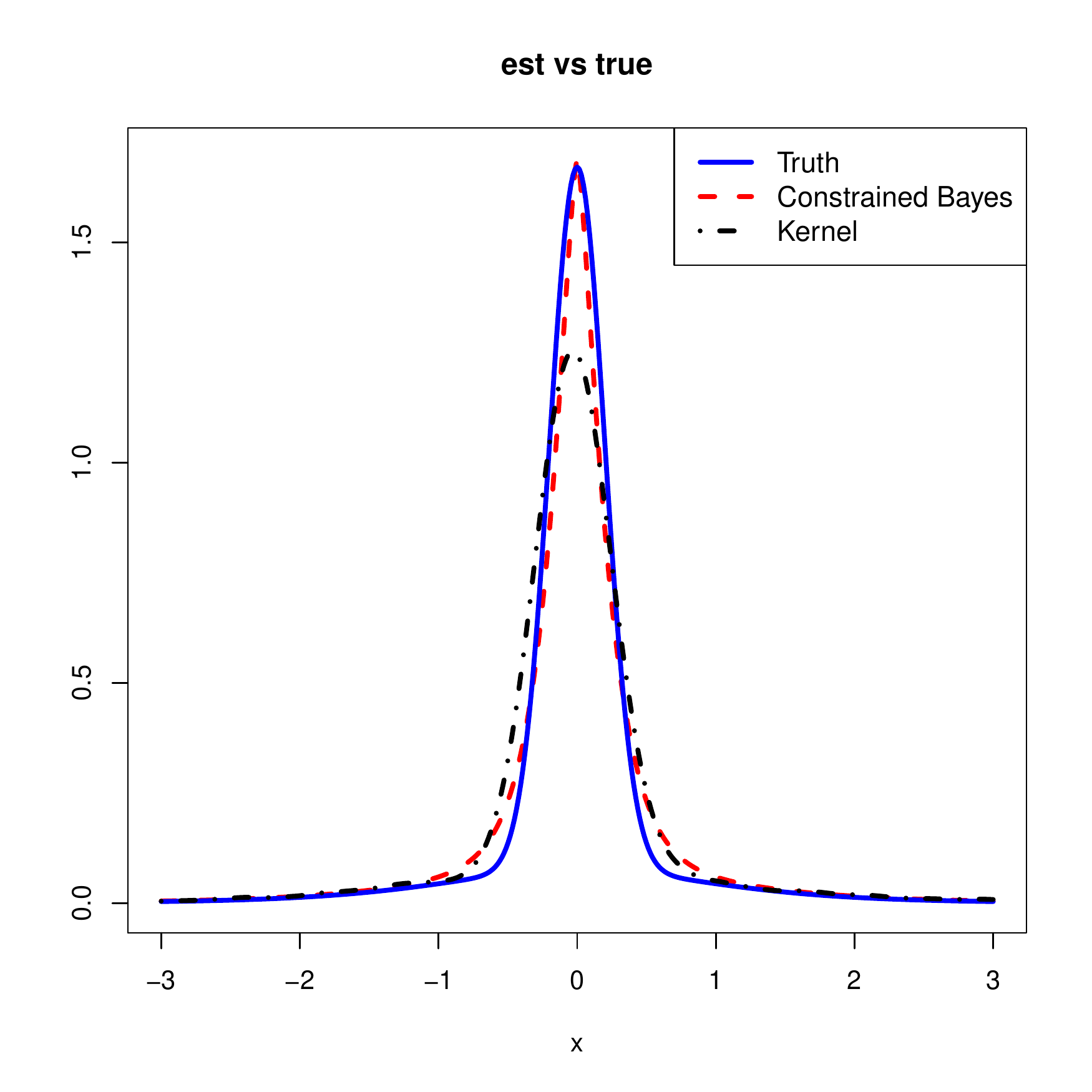}
\vskip -10pt
\caption{\baselineskip=12pt Mean density estimates for the
  heteroscedastic Laplace measurement error simulation of Section \ref{sec5.3} for sample size $n=5000$. Solid blue line is the truth (Truth, a mixture of a t--density
  with 5 degrees of freedom and a Normal density with standard deviation $0.2$), the dashed red line is our Constrained Bayes Deconvolution method (Constrained
  Bayes) and the dash-dotted black line is the deconvoluting kernel density estimator (Kernel).}
\label{sim_fig8}
\end{figure}

\begin{figure}[!ht]
\centering
\includegraphics[height=5in, width=5.5in, trim=1cm 1cm 1cm 1.5cm, clip=true]{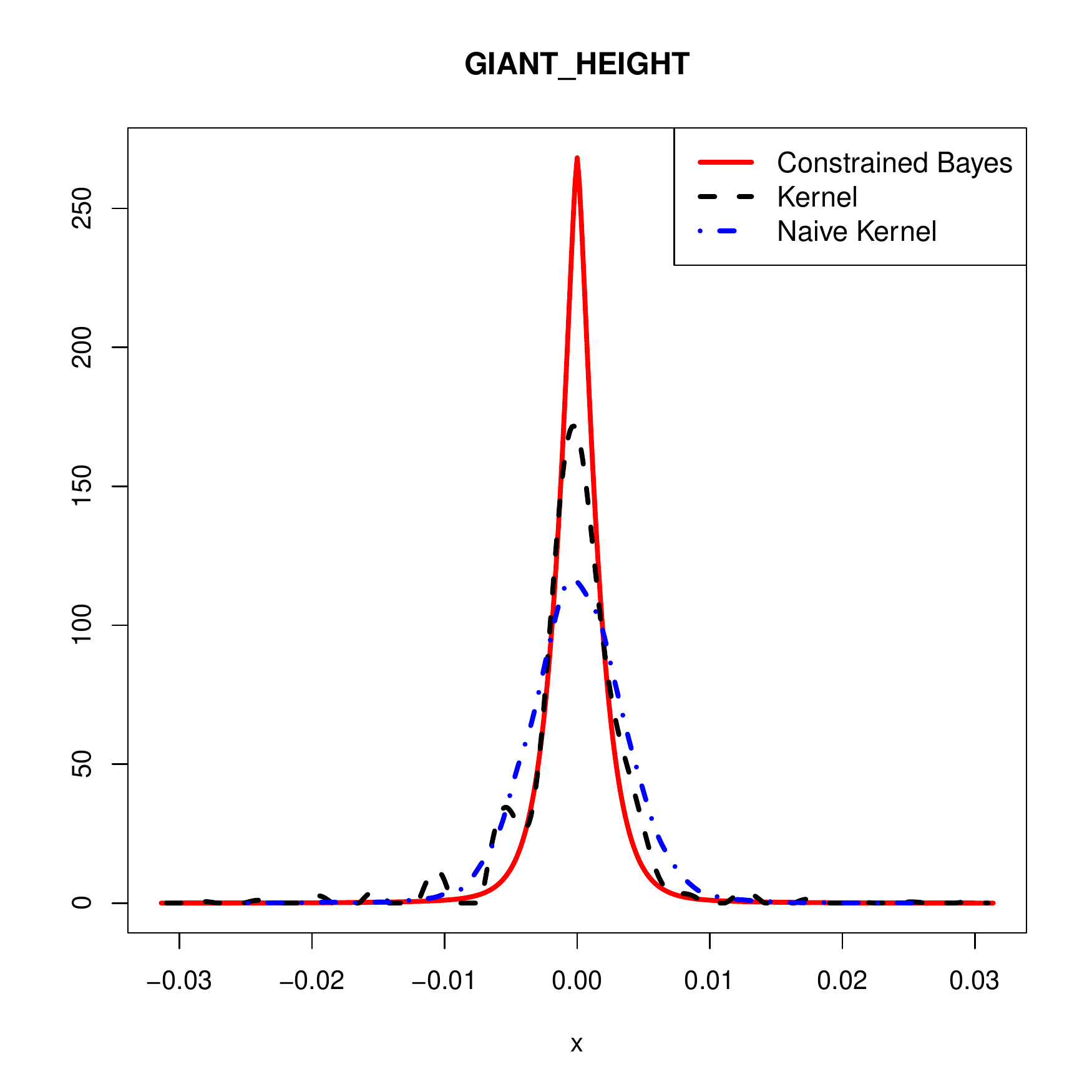}
\vskip -10pt
\caption{\baselineskip=12pt Density estimators for SNP related effect sizes in the GIANT Height data. The solid red line is our Constrained Bayes Deconvolution
method (Constrained Bayes), the dashed black line is the deconvoluting kernel density estimator, but the publicly available R programs are too slow to compute this
and have memory issue on the full data, so we used a 1\% subsample of the data. The dash-dotted blue line is the naive ordinary kernel density estimator ignoring
measurement error. The results for the first and third estimators are similar on the same 1\% subsample are similar to the full data estimates.}
\label{fig3}
\end{figure}

\begin{figure}[!ht]
\centering
\includegraphics[height=5in, width=5.5in, trim=1cm 1cm 1cm 1.5cm, clip=true]{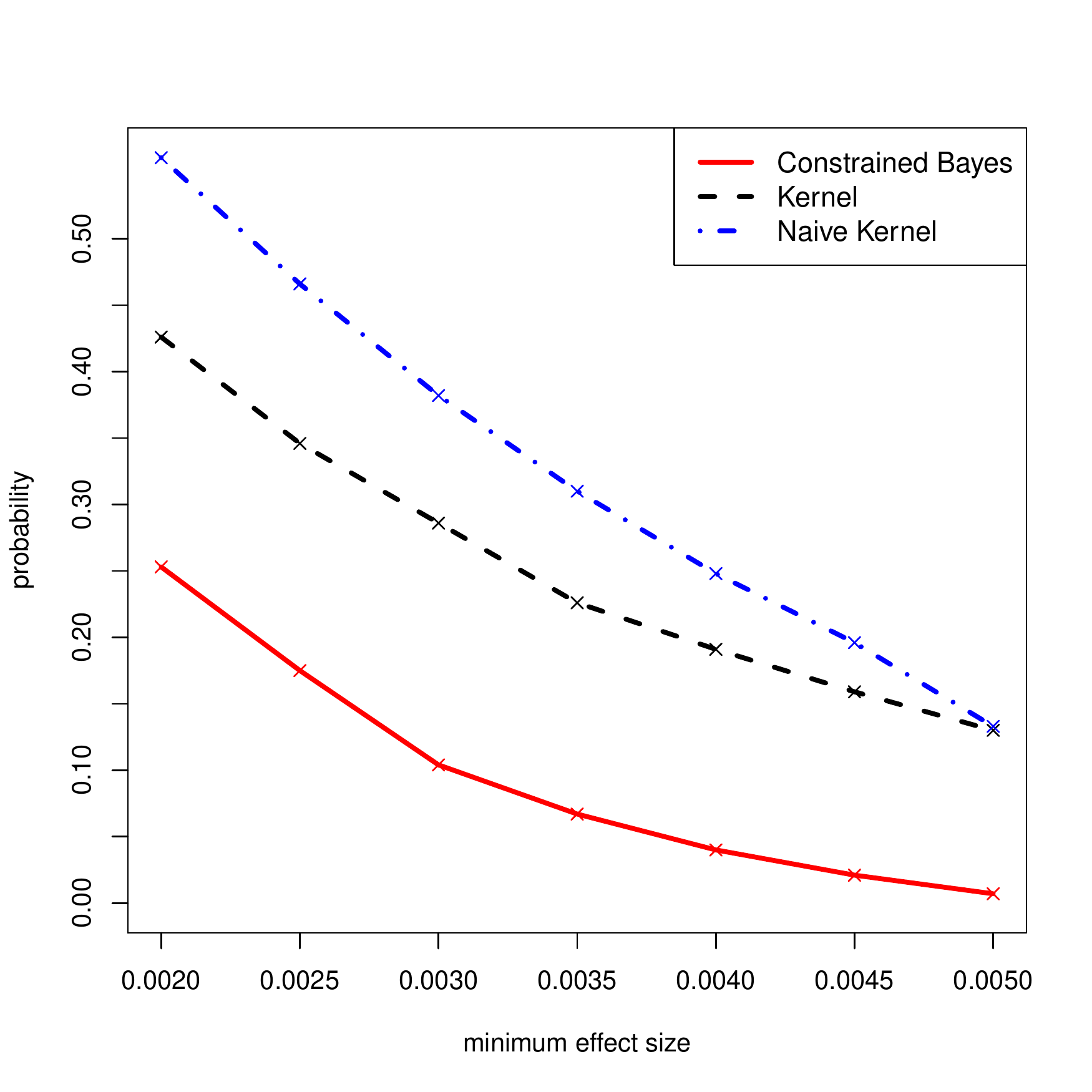}
\vskip -10pt
\caption{\baselineskip=12pt The estimated probability of effect sizes (y-axis) associated with height that the absolute value of effect sizes is greater than the minimum
effect size versus the minimum effect size (x-axis) at some discrete choices as in Table \ref{prob_small_eff_GIANT}. The solid red line is based on our Constrained
Bayes Deconvolution method (Constrained Bayes), the dashed black line is based on the deconvoluting kernel density estimator (Kernel). The dash-dotted blue line is
based on the naive ordinary kernel density estimator (Naive Kernel).}
\label{fig8}
\end{figure}

\clearpage\pagebreak\newpage
\pagestyle{fancy}
\fancyhf{}
\rhead{\bfseries\thepage}
\lhead{\bfseries NOT FOR PUBLICATION SUPPLEMENTARY MATERIAL}
\begin{center}
{\LARGE{\bf Supplementary Material to\\ {\it Nonparametric Bayesian Deconvolution of a Symmetric Unimodal Density, with Application to Genomics}}}
\end{center}

\baselineskip=12pt

\vskip 2mm
\begin{center}
Ya Su \\
Department of Statistics, University of Kentucky, Lexington, KY
40536-0082, U.S.A., ya.su@uky.edu\\
\hskip 5mm\\
Anirban Bhattacharya\\
Department of Statistics, Texas A\&M University, College Station, TX
77843-3143, U.S.A., anirbanb@stat.tamu.edu\\
\hskip 5mm \\
Yan Zhang and Nilanjan Chatterjee\\
Departments of Biostatistics and Oncology, Johns Hopkins University, Baltimore, Maryland 21205, U.S.A., yzhan284@jhu.edu and nchatte2@jhu.edu\\
\hskip 5mm\\
Raymond J. Carroll\\
Department of Statistics, Texas A\&M University, College Station, TX 77843-3143, U.S.A. and School of Mathematical and Physical Sciences, University of
Technology Sydney, Broadway NSW 2007, Australia, carroll@stat.tamu.edu\\
\end{center}

\setcounter{figure}{0}
\setcounter{equation}{0}
\setcounter{page}{1}
\setcounter{table}{1}
\setcounter{section}{0}
\renewcommand{\thefigure}{S.\arabic{figure}}
\renewcommand{\theequation}{S.\arabic{equation}}
\renewcommand{\thesection}{S.\arabic{section}}
\renewcommand{\thesubsection}{S.\arabic{section}.\arabic{subsection}}
\renewcommand{\thepage}{S.\arabic{page}}
\renewcommand{\thetable}{S.\arabic{table}}
\baselineskip=17pt

\section{Overview}\label{sec.S1}

In this supplement, we present a microarray example in Section \ref{sec.S1.1} that has the same structure as that of genome wide association studies (GWAS) in
Section \ref{sec:data} of the main paper. Section \ref{sec.S1.2}
contains some additional simulation results as a complement of the
setup in Section \ref{sec5.3}. In
addition, we also provide our R code that we used in our analyses. This code uses the RCPP package in R to make our calculations feasible for GWAS.

\subsection{Microarray Data}\label{sec.S1.1}

The data we use arise from a complicated experimental design, see
\cite{davidson2004chemopreventive}. A total of 59 male Sprague-Dawley
rats were injected either with saline or the potent carcinogen
Azoxymethane (AOM), and then sacrificed. We measured gene expression values for 8,038 genes, log2 transformed them, and then centered and standardized them.
The treatment (AOM versus saline) was then regressed on the gene expressions, resulting in data similar to that of Section \ref{sec:data}. There were 4514 genes that
had a statistically significant treatment effect with a Bonferroni p-value $<$ 0.05. The effect sizes had a mean of $-0.009$, a skewness of $0.018$ and a kurtosis of
$3.56$. The variabilities of the regression of treatment on the gene expressions had a minimum of 0.008, a maximum of 0.169, and a 5$\th$ percentile of 0.016.

Our Constrained Bayes Deconvolution estimator was applied to the effect sizes associated with treatment. We ran 5000 MCMC iterations under the same
hyperparameters used in the simulation sections. We also implemented the rescaled kernel deconvolution estimator in  \cite{delaigle2008density} based on code
available at Aurore Delaigle's web site. In addition, we computed the naive kernel density estimator which ignores measurement error, available in the R package
KernSmooth. The results are given in Figure \ref{figS.1}. Here we see the same phenomenon seen in the heteroscedastic simulations (Section \ref{sec5.2}) and the
GIANT height data (Section \ref{sec:GIANT}), namely that the Constrained Bayes estimator recognizes more clearly that many of the effect sizes are small, and hence
the density estimate is much more peaked near zero. Another way of writing this is that the kernel methods think there are a more genes with larger effect sizes.

\begin{figure}[!ht]
\centering
\includegraphics[height=5in, width=5.5in, trim=1cm 1cm 1cm 1.5cm, clip=true]{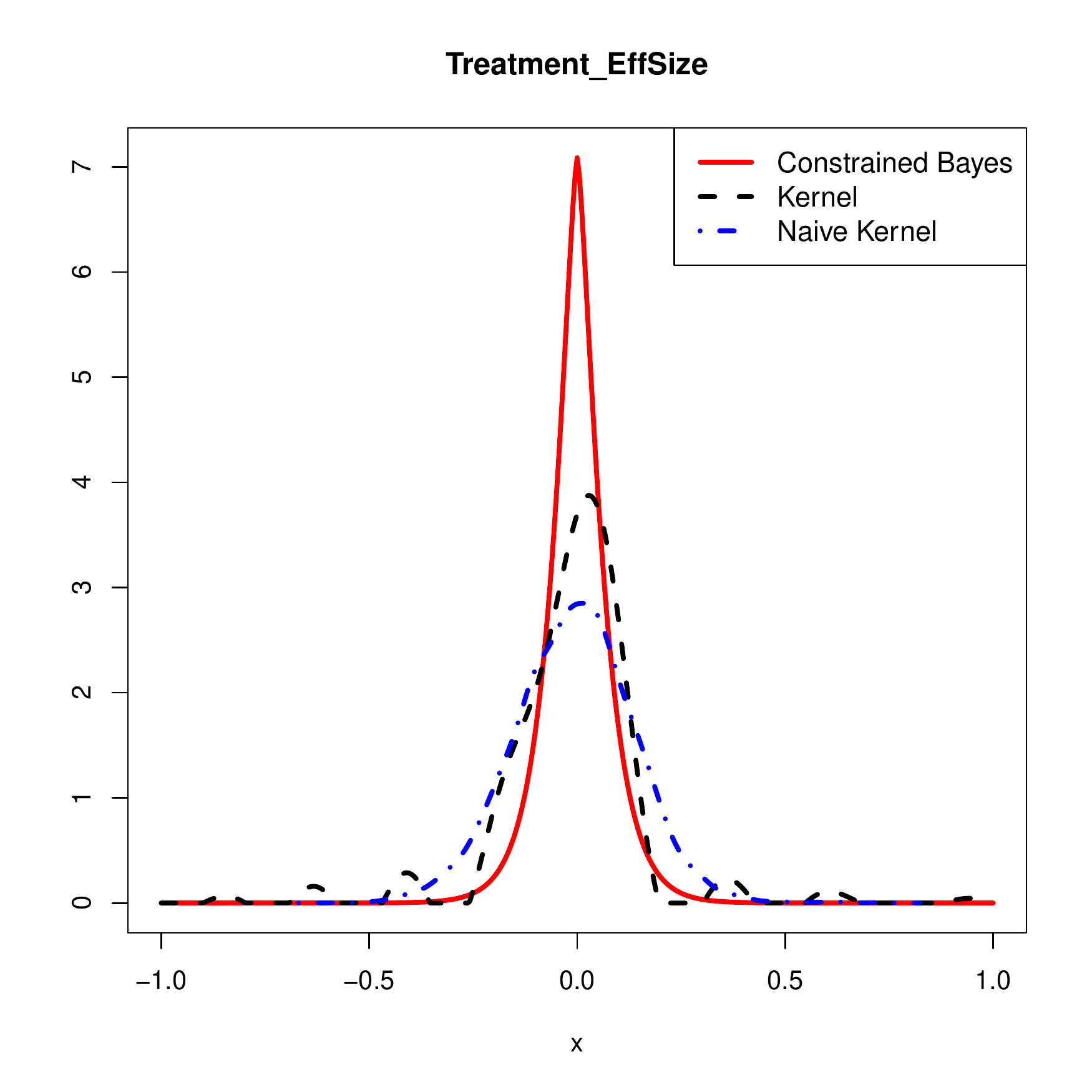}
\vskip -10pt
\caption{\baselineskip=12pt Density estimators for treatment effect sizes in the microarray data of Section \ref{sec.S1.1}. The solid red line is our Constrained
Bayesian method (Constrained Bayes). The dashed black line is the deconvoluting kernel density estimator that recognizes measurement error and potential
heteroscedasticity. The dash-dotted blue line is the naive ordinary kernel density estimator ignoring measurement error.}
\label{figS.1}
\end{figure}

\begin{table}[htbp]
\centering
\begin{tabular}{lcccccc}
\hline\hline
  &\multicolumn{6}{c}{Minimum effect size}\\
  \cmidrule(lr){2-7}
Estimator             & 0.01  & 0.05 & 0.1 & 0.15 & 0.2 & 0.25\\
\hline
Constrained Bayes     & 0.859 & 0.425  & 0.150 & 0.038 & 0.000 & 0.000  \\*[-.60em]
Kernel                & 0.926 & 0.643  & 0.357 & 0.176 & 0.092 & 0.067   \\
\hline\hline
\end{tabular}
\caption{\baselineskip=12pt Comparison of estimated probability of effect sizes associated with treatment that the absolute value of effect sizes is greater than the
minimum effect size under our constrained Bayesian method (Constrained Bayes), the deconvoluting kernel density estimator (Kernel) when rats with multiple arrays
have their expressions averaged, which ends up with 59 observations.}
\label{prob_small_eff_59_Treatment_First_Array}
\end{table}

\begin{figure}[!ht]
\centering
\includegraphics[height=5in, width=5.5in, trim=1cm 1cm 1cm 1.5cm, clip=true]{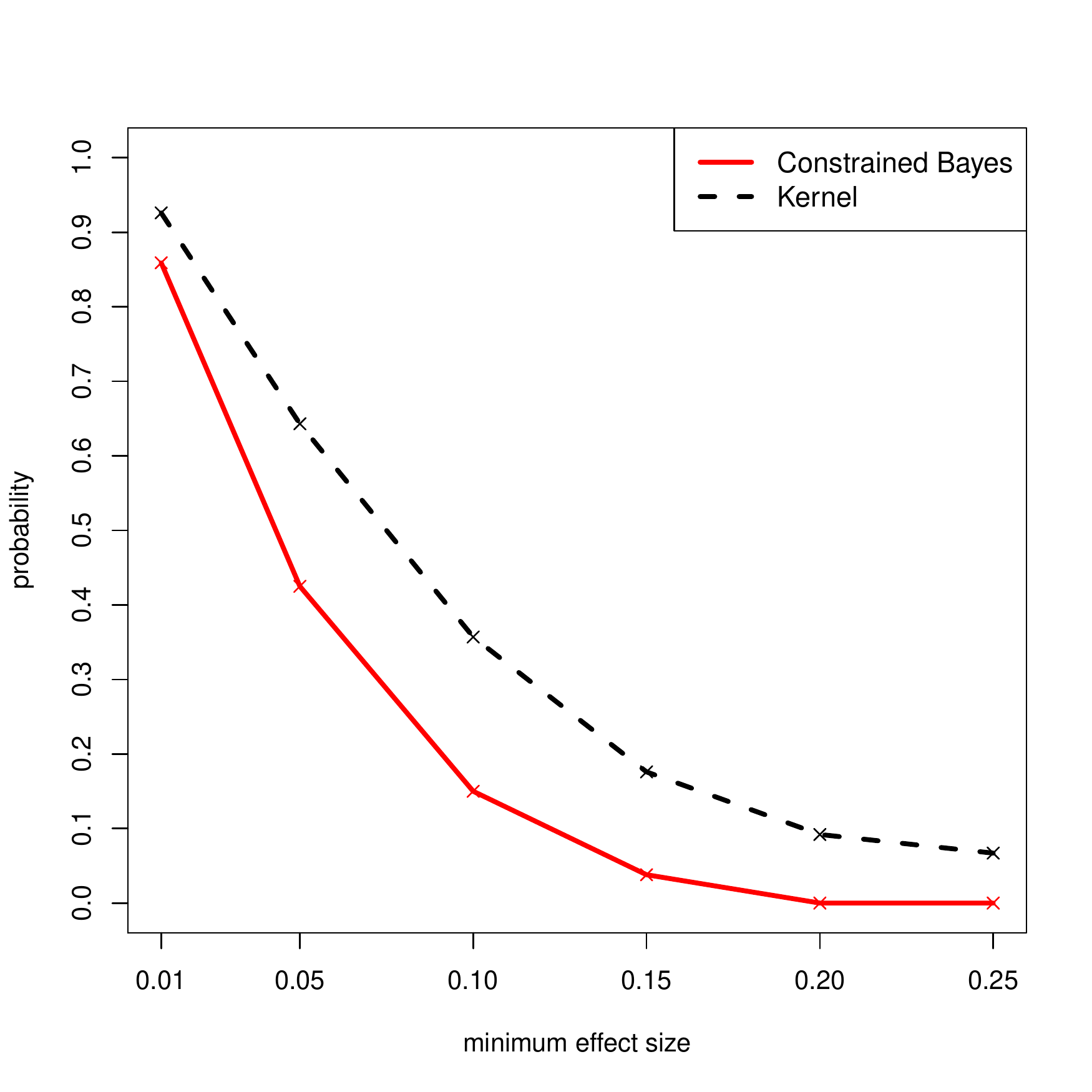}
\vskip -10pt
\caption{\baselineskip=12pt The estimated probability of effect sizes associated with treatment that the absolute value of effect sizes is greater than the minimum effect
size versus the minimum effect size at some discrete choices as in Table \ref{prob_small_eff_59_Treatment_First_Array}. The solid red line is based on our
constrained Bayesian method (Constrained Bayes), the dashed black line is based on the deconvoluting kernel density estimator (Kernel).}
\label{figS.5}
\end{figure}

\subsection{Details about analysis of GIANT Height data}\label{sec.S1.3}
In this section, we report the procedures for selecting independent
SNPs and performing Monte Carlo integration involved in the projection
formula, in Section \ref{sec:GIANT}.

We select a subset of independent SNPs based on linkage disequilibrium (LD) clumping method using PLINK software \citep{purcell2007plink}. LD clumping
typically sorted SNPs according to the importance (p-values) of SNPs, then took the most significant SNPs and removed SNPs that are correlated with this SNP
(squared correlation larger than 0.1) in the window of 1MB base pair
distance. Then it went on with the next most significant SNPs that had
not been removed yet. Using the data set of
  \cite{allen2010hundreds}, the above LD clumping procedure yields
$K = 80349$ independent SNPs.

The projection formula requires an integration with respect to the
density of true effect sizes, $f(\beta)$. Since we do not have a
closed form for $f(\beta)$, we can borrow information from posterior samples of
$\beta_{ij}$, the subscripts $i$ and $j$ indicate the effect sizes
corresponding to the $i$th SNP and in the $j$th MCMC iteration, for $i =
1, \ldots, K$, $j = 1, \ldots, N$ ($K$ and $N$ represent the total
number of SNPs and MCMC iterations). The following steps
are performed to complete the calculations for predicting the expected
number of significant SNPs:
\begin{enumerate}
\item Hypothesize a new sample size $n_{\text{new}}$.
\item To reduce the correlation caused by MCMC chains, we adopt an
  aggressive thinning at every $50$th iteration. For our analysis of
  Height data, the
  original MCMC chain contains $50000$ iterations (burn-ins excluded),
  hence $N = 1000$.
\item For any fixed $j$, compute the expected number of significant SNPs, $\sum_{i=1}^{K}  \mbox{pow}_{\sigma,
    \alpha}(\beta_{ij})$, where $\mbox{pow}_{\sigma,
    \alpha}(\beta) = 1 - \Phi(z_{\alpha/2} -
  n_{\text{new}}^{1/2}\sigma^{-1}\beta) + \Phi(-z_{\alpha/2} -
  n_{\text{new}}^{1/2}\sigma^{-1}\beta)$.
\item Repeat Step 3 for $j = 1,\ldots,N$ times. We can obtain posterior samples
  of the predicted values, and thus, Table \ref{projection_GIANT}.
\end{enumerate}

\subsection{Additional Simulation: The Distribution of $X$ has a Tight Peak Around Zero}\label{sec.S1.2}

We changed the data generating model in Section \ref{sec5.3} to
$\sigma_{00} = 0.1$. Specifically, We implement a mixing of a
$\Normal(0, \sigma_{00}^2)$ and a $t$-distribution with $5$ degrees of
freedom for the second component, with mixing probabilities $0.8$ and
$0.2$ respectively. We choose the small value $\sigma_{00} = 0.1$ so
that the mixing density has a even sharper peak around zero compared
to Section \ref{sec5.3}. The additional simulation has only been implemented for
normally-distributed error. A similar pattern should be expected when
the error distribution is Laplace based on the existing numerical results in \Section
\ref{sec5.3}.

We first consider a homoscedastic error setup, where $\sigma_i^2 = 0.6^2$ as in Section \ref{sec5.3}. See Figure \ref{fig6} for the result of the averaged density over
$100$ simulations in this setting with $n = 5000$. The numerical comparison for our Constrained Bayes method and the kernel method is given in Table
\ref{S1:tab6}.

\begin{table}[htbp]
  \centering
\begin{tabular}  {llcc}
\hline\hline
              && Constrained & \\*[-.60em]
    \( n \) & & Bayes & Kernel \\
    \hline
    1000 & IAE & 0.730 (0.041) & 1.069 (0.082) \\*[-.60em]
    & ISE & 1.159 (0.059) & 1.068 (0.061) \\*[-.60em]
    & Exceedance & 0.235 (0.028) & 0.415 (0.052) \\
    5000 & IAE & 0.570 (0.041) & 1.018 (0.053) \\*[-.60em]
    & ISE & 0.916 (0.065) & 1.035 (0.040) \\*[-.60em]
    & Exceedance & 0.147 (0.019) & 0.382 (0.030) \\
    10000 & IAE & 0.508 (0.036) & 1.006 (0.047) \\*[-.60em]
    & ISE & 0.820 (0.061) & 1.031 (0.035) \\*[-.60em]
    & Exceedance & 0.120 (0.015) & 0.375 (0.025) \\
    15000 & IAE & 0.474 (0.046) & 0.998 (0.046) \\*[-.60em]
    & ISE & 0.767 (0.078) & 1.023 (0.032) \\*[-.60em]
    & Exceedance & 0.107 (0.015) & 0.369 (0.025) \\
  \hline\hline
  \end{tabular}
  \caption{\baselineskip=12pt
   Comparison of our constrained Bayesian method (Constrained Bayes), the deconvoluting kernel density estimator (Kernel). This is in the first case of Section
   \ref{sec.S1.2}, when the target density is a mixture of t-density with 5 degrees of freedom and a Normal density with standard deviation $0.1$, and when the
   measurement errors are homoscedastic. The sample size is $n$, IAE is integrated absolute error, ISE is integrated squared error and Exceedance is the absolute
   difference between the exceedance probability under the estimated density and that under the true density. Numbers in parentheses are standard errors.}
  \label{S1:tab6}
\end{table}

We implement the heteroscedastic and select $\sigma_i^2$ as in Section \ref{sec5.3}, specifically, $\sigma_i^2 = (0.75 + X_i/4)^2$. Figure \ref{fig7} shows the
estimated density averaged over $100$ simulated data sets with $n = 5000$. The numerical comparison for our Constrained Bayes method and the kernel method is
given in Table \ref{S1:tab7}.

    \begin{table}[htbp]
  \centering
\begin{tabular}  {llccc}
\hline\hline
              && Constrained & \\*[-.60em]
    \( n \) & & Bayes & Kernel \\
    \hline
    1000 & IAE & 0.847 (0.053) & 1.175 (0.033) \\*[-.60em]
    & ISE & 1.307 (0.063) & 1.138 (0.021) \\*[-.60em]
    & Exceedance & 0.311 (0.033) & 0.497 (0.026) \\
    5000 & IAE & 0.668 (0.044) & 1.159 (0.016) \\*[-.60em]
    & ISE & 1.067 (0.062) & 1.126 (0.010) \\*[-.60em]
    & Exceedance & 0.216 (0.023) & 0.480 (0.012) \\
    10000 & IAE & 0.578 (0.040) & 1.154 (0.015) \\*[-.60em]
    & ISE & 0.935 (0.061) & 1.123 (0.008) \\*[-.60em]
    & Exceedance & 0.177 (0.018) & 0.473 (0.010) \\
    15000 & IAE & 0.532 (0.045) & 1.149 (0.012) \\*[-.60em]
    & ISE & 0.862 (0.074) & 1.119 (0.006) \\*[-.60em]
    & Exceedance & 0.159 (0.017) & 0.469 (0.009) \\
    \hline\hline
  \end{tabular}
\caption{\baselineskip=12pt
   Comparison of our constrained Bayesian method (Constrained Bayes),
   the deconvoluting kernel density estimator (Kernel). This is in the
   second case of Section \ref{sec.S1.2}, when the target density is a
   mixture of t-density with 5 degrees of freedom and a Normal density
   with standard deviation $0.1$, and when the measurement errors are
   heteroscedastic. The sample size is $n$, IAE is integrated
   absolute error, ISE is integrated squared error and Exceedance is the absolute difference between the exceedance probability under the estimated density and that
   under the true density. Numbers in parentheses are standard errors.}
    \label{S1:tab7}
  \end{table}

\begin{figure}[!ht]
\centering
\includegraphics[height=5in, width=5.5in, trim=1cm 1cm 1cm 1.5cm, clip=true]{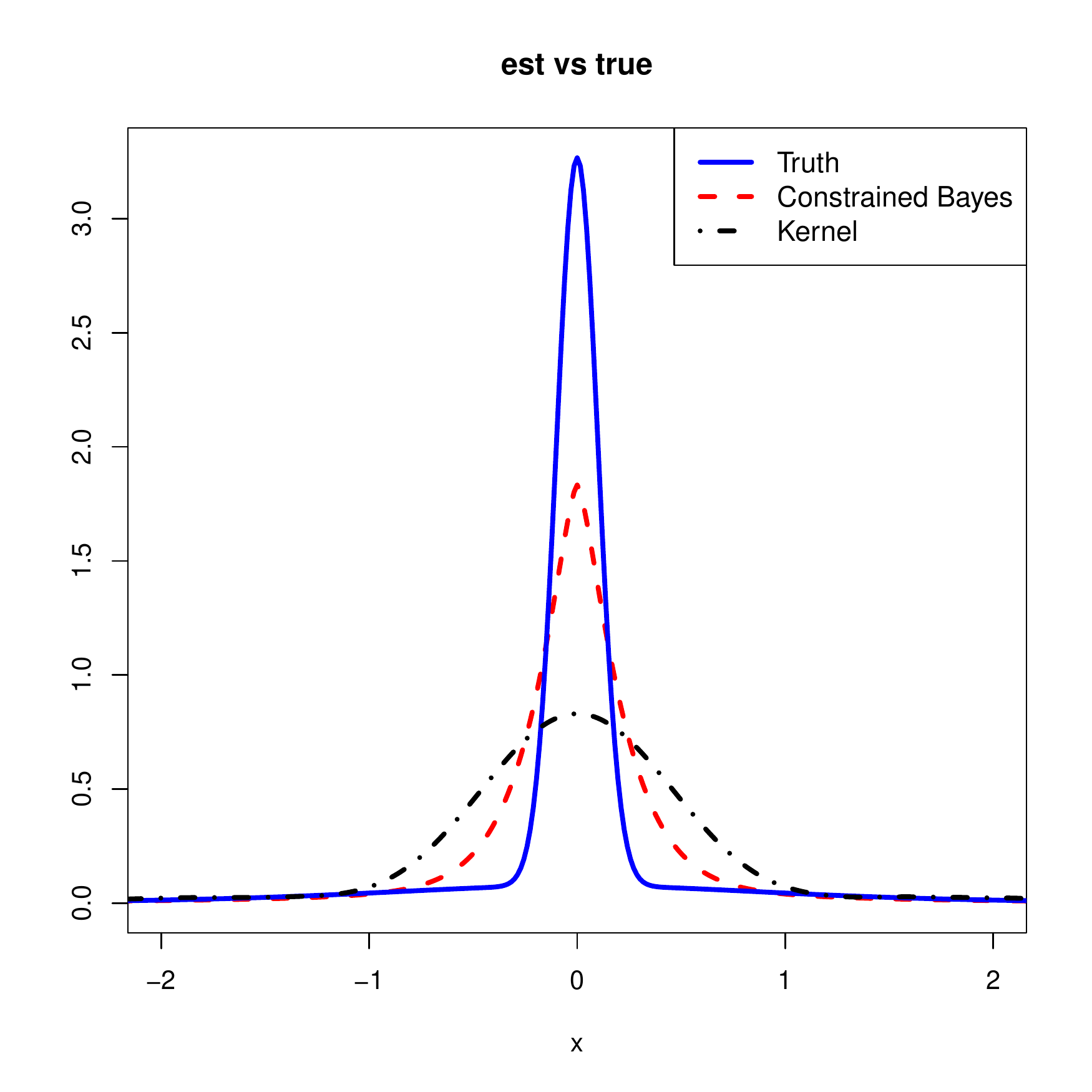}
\vskip -10pt
\caption{\baselineskip=12pt Mean density estimates for the homoscedastic simulation of Section \ref{sec.S1.2} for sample size $n=5000$. Solid blue line is the truth
(Truth, a mixture of a t--density with 5 degrees of freedom and a Normal density with standard deviation $0.1$), the dashed red line is our constrained Bayesian
method (Constrained Bayes) and the dash-dotted black line is the deconvoluting kernel density estimator (Kernel).}
\label{fig6}
\end{figure}

\begin{figure}[!ht]
\centering
\includegraphics[height=5in, width=5.5in, trim=1cm 1cm 1cm 1.5cm, clip=true]{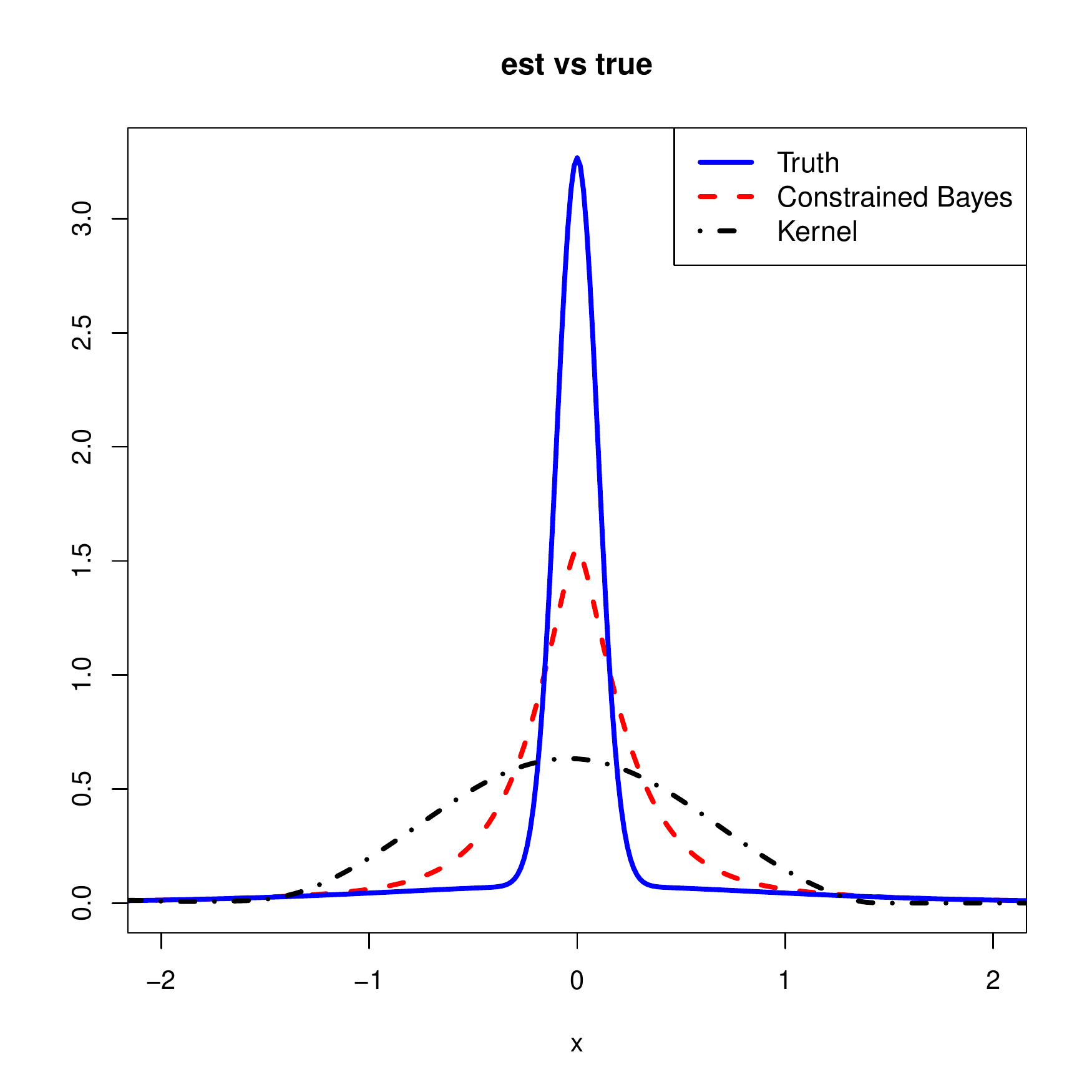}
\vskip -10pt
\caption{\baselineskip=12pt Mean density estimates for the heteroscedastic simulation of Section \ref{sec.S1.2} for sample size $n=5000$. Solid blue line is the truth
(Truth, a mixture of a t--density with 5 degrees of freedom and a Normal density with standard deviation $0.1$), the dashed red line is our constrained Bayesian
method (Constrained Bayes) and the dash-dotted black line is the deconvoluting kernel density estimator (Kernel).}
\label{fig7}
\end{figure}

\end{document}